\documentclass[11pt]{article}

\usepackage[margin=1in]{geometry}

\usepackage{times}

\usepackage{soul}
\usepackage{url}
\usepackage[hidelinks]{hyperref}
\usepackage[utf8]{inputenc}
\usepackage[small]{caption}
\usepackage{graphicx}
\usepackage{amsmath,amssymb,amsthm}
\usepackage{amsfonts}
\usepackage{booktabs}
\usepackage{color}
\usepackage{xfrac}
\usepackage{multirow}

\newcommand{\MMS}{\mathsf{MMS}}

\newcommand{\bX}{\mathbf{X}}

\usepackage[ruled,linesnumbered,vlined]{algorithm2e}
\usepackage[capitalize]{cleveref}

\newenvironment{proofof}[1]{{\vspace*{5pt} \noindent\bf Proof of #1:  }}{\hfill\rule{2mm}{2mm}\vspace*{5pt}}

\newtheorem{claim}{Claim}[section]
\newtheorem{theorem}{Theorem}[section]
\newtheorem{corollary}{Corollary}[theorem]
\newtheorem{lemma}[theorem]{Lemma}
\newtheorem{observation}[theorem]{Observation}

\newtheorem{definition}[theorem]{Definition}

\DeclareMathOperator*{\argmax}{argmax}
\DeclareMathOperator*{\argmin}{argmin}

\title{Multi-agent Online Scheduling: MMS Allocations for Indivisible Items}

\author{Shengwei Zhou
	\thanks{IOTSC, University of Macau, China. \{yc17423,yb97439,xiaoweiwu\}@um.edu.mo}
	\and
	Rufan Bai $^*$
	\and
	Xiaowei Wu $^*$
}

\begin{document}
	
	\maketitle
	
	\begin{abstract}
		We consider the problem of fairly allocating a sequence of indivisible items that arrive online in an arbitrary order to a group of $n$ agents with additive normalized valuation functions.
		We consider both the allocation of goods and chores, and propose algorithms for approximating maximin share (MMS) allocations.
		When agents have identical valuation functions the problem coincides with the semi-online machine covering problem (when items are goods) and load balancing problem (when items are chores), for both of which optimal competitive ratios have been achieved.
		In this paper we consider the case when agents have general additive valuation functions.
		For the allocation of goods we show that no competitive algorithm exists even when there are only three agents and propose an optimal $0.5$-competitive algorithm for the case of two agents.
		For the allocation of chores we propose a $(2-1/n)$-competitive algorithm for $n\geq 3$ agents and a $\sqrt{2}\approx 1.414$-competitive algorithm for two agents.
		Additionally, we show that no algorithm can do better than $15/11\approx 1.364$-competitive for two agents.
	\end{abstract}

\section{Introduction} \label{sec:introduction}

Traditional machine learning algorithms usually focus on global objectives such as efficiency or maximizing profit and have no guarantee of fairness between individuals.
As learning decision-making is increasingly involved in our daily life, the problem of algorithm bias has received increasing attention.
Motivated by several real-world problems in which decision processes impact human beings who must be treated fairly and unbiasedly, there is increasing attention on fair learning algorithms~\cite{conf/icml/ChenFLM19,conf/icml/BackursIOSVW19,conf/icml/LiLSWW21}.
We focus on the online scheduling problem where jobs arrive online and each incoming job should be assigned by the algorithm immediately.
Contrast to the classic online scheduling models that focus on allocation to machines and object on optimizing global objectives, e.g., maximizing the minimum load~\cite{journals/tcs/TanW07,conf/nips/MaxMDM22,journals/tcs/Cheng22a}, or minimizing the maximum load~\cite{journals/scheduling/KellererKG15}, we consider the allocation to agents and study the online scheduling problem in a multi-agent perspective.
The problem coincides with the online maximin share (MMS) allocation problem proposed by~\cite{journals/corr/abs-2208-08782}, which falls into the class of fair allocation problems.

In the fair allocation problem, there is a set $M$ of $m$ indivisible items (jobs) and a group $N$ of $n$ (heterogeneous) agents, where each agent $i\in N$ has a valuation function $v_i$ on the items.
For indivisible items, each item $e\in M$ must be allocated to exactly one of the agents in $N$.
Therefore each allocation corresponds to a partitioning of the items into $n$ bundles $(X_1,\ldots,X_n)$, where agent $i\in N$ receives bundle $X_i$.
When agents have positive values on the items, we call the items \emph{goods}, e.g., consider the allocation of gifts to kids; when agents have negative values on the items, we call the items \emph{chores}, e.g., when allocating tasks to workers.
In this paper, we study both the allocation of goods and chores.	
For the case of chores, we assume that agents have positive \emph{costs} on the items and refer to the valuation function of agent $i$ as a cost function $c_i: 2^N \to \mathbb{R}^+$.

Different fairness notions have been proposed to measure how fair an allocation is, e.g., the envy-freeness (EF)~\cite{foley1967resource,conf/sigecom/LiptonMMS04,journals/teco/CaragiannisKMPS19}, proportionality (PROP)~\cite{steihaus1948problem} and maximin fair share (MMS)~\cite{conf/bqgt/Budish10}.
In this paper, we focus on the fairness notion of MMS.
Informally speaking, the MMS value $\MMS_i$ of an agent $i\in N$ is the best she can guarantee if she gets to partition items into $n$ bundles but is the last agent to pick a bundle.
An allocation is called MMS if every agent receives a bundle with objective no worse than her MMS value.
For the case of goods, that means $v_i(X_i) \geq \MMS_i$ for all $i\in N$; for the case of chores, that means $c_i(X_i) \leq \MMS_i$ for all $i\in N$.
%
%
Similar to several works on online scheduling~\cite{conf/waoa/EbenlendrNSW05,conf/escape/WuTY07,angelelli2007semi,journals/tcs/ChengKK05,journals/scheduling/LeeL13}, we consider the setting of semi-online in which algorithms are given partial information before items (jobs) arrive.
We assume the sum of the total value of items (the total processing time of jobs) is known, which can be modeled as all valuation functions (resp. cost functions) are normalized, i.e. $v_i(M) = n$ (resp. $c_i(M) =n$) for all $i\in N$.
Hence for allocating goods, the approximation ratio with respect to MMS is at most $1$ while for chores it is at least $1$.
The partial information can be considered {\em learned information} based on the historical data.
For example, in the food bank problem proposed by Walsh~\cite{conf/ki/Walsh14}, the food bank has to distribute donated food to charities in an online manner.
While the amount of food is unknown on an hourly or daily basis, the total amount in a fixed period (e.g. a week) can often be accurately estimated based on historical data.
To name another example, consider a manufacturing company, e.g. Foxconn, that receives online manufacturing orders and needs to assign orders to different working units in a fair way.
Again, while the daily order volume may fluctuate greatly, the total volume in a month is often stable and predictable.

The problem of online approximation of the MMS allocations for agents with normalized identical valuations coincides with the {\em semi-online machine covering} problems~\cite{conf/waoa/EbenlendrNSW05,conf/escape/WuTY07} when items are goods, and the {\em semi-online load balancing} problem~\cite{angelelli2007semi,journals/tcs/ChengKK05,journals/scheduling/LeeL13} when items are chores, in the research field of online scheduling (where the items are jobs and agents are machines).
Specifically, for allocating goods, the common MMS of agents corresponds to the minimum load of the machines in the optimal scheduling; for chores, it corresponds to the maximum load of the machines, e.g., the makespan, in the optimal scheduling.
For $n=2$ agents, Kellerer et al.~\cite{journals/orl/KellererKST97} present a $2/3$-competitive algorithm for goods and a $4/3$-competitive algorithm for chores and show that these competitive ratios are the best possible.
For $n\geq 3$ agents, Tan and Wu~\cite{journals/tcs/TanW07} propose a $1/(n-1)$-competitive algorithm and show that it is optimal for the allocation of goods; Kellerer et al.~\cite{journals/scheduling/KellererKG15} give a $1.585$-competitive algorithm for the allocation of chores, which is also optimal due to the lower bound by	Albers and Hellwig~\cite{journals/tcs/AlbersH12}.

\subsection{Our Results}

In this paper, we consider both the allocation of goods and chores and present upper and lower bounds for approximating MMS allocations.
As in~\cite{conf/aaai/GkatzelisPT21,journals/mansci/BogomolnaiaMS22,conf/aaai/Barman0M22}, we assume that the valuation functions are normalized\footnote{In Appendix~\ref{sec:justification}, we provide some justifications for this assumption, showing that without this assumption (1) for the case of goods, the competitive ratio is arbitrarily bad; (2) for the case of chores, the problem becomes strictly harder.}, e.g., $v_i(M) = n$ for all $i\in N$ for goods and $c_i(M) = n$ for all $i\in N$ for chores.
We refer to an online algorithm as \emph{$r$-competitive} if, for any online instance, the allocation returned by the algorithm is always $r$-approximate MMS.
We only consider deterministic algorithms in this paper.

Due to the unknown future, the online setting brings many challenges to the fair allocation problem and most of the classic algorithms cease to work.
There are two main difficulties in designing online algorithms: (1) the irrevocable decision-making, and (2) the arbitrary arrival order of items.
For example, the \emph{envy-cycle elimination}~\cite{conf/sigecom/LiptonMMS04} algorithms heavily rely on exchanging items among agents to improve the allocation, which is not allowed in the online setting, as the allocation decisions are irrevocable.
Other classic algorithms for approximating MMS allocations~\cite{aziz2022approximate,conf/aaai/AzizRSW17,conf/sigecom/HuangL21} are built on the reduction to identical ordering instances and allocate items in decreasing order of values/costs.
Unfortunately in the online setting, the algorithm has no control over the arrival order of items and thus all these algorithms fail to work.

To get around these difficulties, we combine the ideas from the fair allocation field with that from the online scheduling field and propose (deterministic) competitive online algorithms that work against adaptive adversaries, for both goods and chores.
For the allocation of goods, we show that no algorithm has a competitive ratio strictly larger than $0$ with respect to MMS, even when there are only three agents.
In contrast, we show that competitive algorithms exist for $n=2$ agents by proposing a $0.5$-competitive algorithm, and show that this is optimal for any online algorithms.
We also present a $0.5$-competitive algorithm for general number of agents under the assumption that the items arrive in the order from the most valuable to the least valuable.
We further consider the small goods instances in which the value of each item (to each agent) is bounded by some $\alpha < 1$ and present a $(1-\alpha)$-competitive algorithm for general number of agents.
Then we turn to the allocation of chores and propose a $(2-1/n)$-competitive algorithm for $n$ agents.
For the case of $n=2$ agents, this gives a $1.5$-competitive algorithm.
We further improve this competitive ratio to $\sqrt{2}\approx 1.414$ by giving another efficient algorithm and provide a hard instance showing that no online algorithm can do better than $15/11\approx 1.364$ competitive for $n=2$.
Moreover, we consider the case when items arrive in the order from the most costly to the least costly, and present an algorithm that is $5/3$-competitive for general number of agents.
Finally, we consider small chores instances in which the cost of each item (to each agent) is bounded by some $\alpha < 1$ and demonstrate the existence of $(1+\alpha)$-competitive algorithm for general number of agents.
For the case of two agents, we improve this competitive ratio to $\sqrt{\alpha^2-4\alpha+5}+\alpha-1$ for small chores instances.
We summarize the upper and lower bounds on the competitive ratios in Table~\ref{table:results}.

\begin{table*}[htb]
\centering
\caption{Summary of results, where Lower and Upper stand for Lower Bound and Upper Bound, respectively.}
\begin{tabular}{|c||cc|cc|} 
    \hline
    & \multicolumn{2}{c|}{Goods}  & \multicolumn{2}{c|}{Chores}  \\ \hline
    & Lower & Upper  & Lower  & Upper \\ \hline \hline
    $n\geq 3$ & {$0$}   & {$0$}  & {$1.585$~\cite{journals/tcs/AlbersH12}} & {$2-1/n$} \\\hline
    $n=2$ & {$0.5$}  & {$0.5$}  & {$15/11$} & {$\sqrt{2}$}  \\ \hline
\end{tabular}
\label{table:results}
\end{table*}

\subsection{Other Related Works} 
In the traditional offline fair allocation problem, it has been shown that MMS allocations are not guaranteed to exist for goods (by Kurokawa et al.~\cite{journals/jacm/KurokawaPW18}) and for chores (by Aziz et al.~\cite{conf/aaai/AzizRSW17}).
Hence, many research focus on the computation of approximately MMS allocations, e.g., for goods ~\cite{journals/jacm/KurokawaPW18,journals/mor/GhodsiHSSY21,journals/ai/GargT21} and chores~\cite{conf/aaai/AzizRSW17,conf/sigecom/BarmanM17,conf/sigecom/HuangL21}.
The state-of-the-art approximation ratio is $(\frac{3}{4}+\min\{\frac{1}{36},\frac{3}{16n-4}\})$ for goods by Akrami et al.~\cite{journals/corr/abs-2303-16788} and $13/11$ for chores by Huang and Segal-Halevi~\cite{journals/corr/abs-2302-04581}.
Recently, Feige et al.~\cite{conf/wine/FeigeST21} show that no algorithm can achieve approximation ratios larger than $39/40$ for goods and smaller than $44/43$ for chores.

Similar to the literature on online scheduling, classic models for online allocation problems often focus on optimizing a global objective~\cite{conf/soda/BanerjeeGGJ22,conf/aaai/Barman0M22,conf/sagt/KawaseS22,journals/scheduling/KellererKG15}.
Recently, inspired by many real-work applications, e.g., the allocation of food to charities in the food bank problem~\cite{conf/ki/Walsh14,conf/aaai/Walsh15,conf/ijcai/AleksandrovAGW15}, and the allocation of tasks to workers in scheduling problems~\cite{journals/orl/KellererKST97,journals/tcs/ChengKK05}, some researchers turn to study the online fair allocation problem.
%
%
In contrast to the rich literature on this problem for divisible items~\cite{journals/jair/KashPS14,conf/ijcai/LiLL18,journals/mansci/BogomolnaiaMS22}, the case of indivisible items is much less well-studied.
To name a few, for allocating indivisible goods, Benade et al.~\cite{conf/sigecom/BenadeKPP18} propose algorithms for minimizing the expected maximum envy among agents;
He et al.~\cite{conf/ijcai/HePPZ19} study the problem of minimizing the number of reallocations to ensure an EF1 allocation;
Zeng and Psomas~\cite{conf/sigecom/ZengP20} study the tradeoff between fairness and efficiency when the values of items are randomly drawn from a distribution.

For a more comprehensive review of other works on fair allocation and online fair allocation problems, please refer to the survey by Amanatidis et al.~\cite{journals/corr/abs-2208-08782} and Aleksandrov and Walsh~\cite{conf/aaai/AleksandrovW20}, respectively.

	\section{Preliminaries} \label{sec:preliminaries}
	
	We consider how to fairly allocate a set of $m$ indivisible goods (or chores) $M$ to a group of $n$ agents $N$, where items $M$ arrive online in an arbitrary order and agents $N$ are offline.
	When items are goods, each agent $i\in N$ has a value $v_i(e)\geq 0$ on each item $e\in M$.
	That is, agent $i\in N$ has an additive valuation function $v_i: 2^M \rightarrow \mathbb{R}^+ \cup \{0\}$ that assigns a positive value $v_i(S) = \sum_{e\in S} v_i(e)$ to any subset of items $S\subseteq M$, and the agents would like to maximize their values.
	When items are chores, we use $c_i(e)\geq 0$ to denote the cost agent $i$ has on item $e$, where $c_i$ is the cost function.
	We have $c_i(S) = \sum_{e\in S} c_i(e)$ for all $S\subseteq E$, and the agents would like to minimize their costs.
	We assume that the valuation and cost functions are normalized, i.e., $v_i(M) = c_i(M) = n$ for all $i\in N$.
	An allocation is represented by an $n$-partition $\bX = (X_1,\cdots,X_n)$ of the items, where $X_i \cap X_j = \emptyset$ for all $i \neq j$ and $\cup_{i\in N} X_i = M$.
	In allocation $\bX$, agent $i\in N$ receives bundle $X_i$.
	Given any set $X\subseteq M$ and $e\in M$, we use $X+e$ and $X-e$ to denote $X\cup\{e\}$ and $X\setminus\{e\}$, respectively.
	
	\begin{definition}[MMS for Goods]
		Let $\Pi(M)$ be the set of all n-partition of $M$.
		For the allocation of goods, for all agent $i\in N$, her maximin share (MMS) is defined as:
		\begin{equation*}
			\MMS_i = \max_{X\in \Pi(M)} \min_{j\in N} \{v_i(X_j)\}.
		\end{equation*}
		For any $\alpha \in [0,1]$, allocation $\bX$ is $\alpha$-approximate maximin share fair ($\alpha$-MMS) if $v_i(X_i)\geq \alpha\cdot\MMS_i$ holds for all $i\in N$.
		When $\alpha = 1$, the allocation $\bX$ is MMS.
	\end{definition}    
	
	\begin{definition}[MMS for Chores]
		Let $\Pi(M)$ be the set of all n-partition of $M$.
		For the allocation of chores, for all agent $i\in N$, her maximin share (MMS) is defined as:
		\begin{equation*}
			\MMS_i = \min_{X\in \Pi(M)} \max_{j\in N} \{c_i(X_j)\}.
		\end{equation*}
		For any $\alpha\geq 1$, allocation $\bX$ is $\alpha$-approximate maximin share fair ($\alpha$-MMS) if $c_i(X_i)\leq \alpha\cdot\MMS_i$ holds for any $i\in N$.
		When $\alpha = 1$, the allocation $\bX$ is MMS.
	\end{definition}

	\paragraph{Online Setting.}
	We use $e_1,e_2,\ldots,e_m$ to index the items $M$ in the order they arrive.
	We consider the adversarial setting in which the adversary designs the instance and decides the arrival order of items.
	Moreover, since our algorithms are deterministic, the adversary can be adaptive, i.e., the adversary can design the value or cost of the online item depending on the previous decisions by the algorithm.
	The algorithm does not know $m$ (the number of items), but knows the number of agents $n$ and that $v_i(M) = n$ (for goods) or $c_i(M) = n$ (for chores).
	The value $v_i(e_j)$ (resp. cost $c_i(e_j)$) of item $e_j \in M$ is revealed for all $i\in N$ upon the arrival of $e_j$.
	Then the online algorithm must make an irrevocable decision on which agent $i\in N$ this item $e_j$ should be assigned to.
	In other words, no reallocations of items are allowed.
	The performance of the algorithm is measured by the competitive ratio, which is the worst approximation guarantee (with respect to MMS) of the final allocation $\bX$ over all online instances.
	
	\section{Allocation of Goods} \label{sec:goods}
	We first consider the allocation of goods and provide upper and lower bounds for the competitive ratio of online algorithms for approximating MMS allocations.
	Recall that for the allocation of goods, the larger the competitive ratio the better, and thus upper bound corresponds to impossibility results while the lower bound corresponds to algorithmic results.
	As introduced, when agents have identical valuation functions, optimal competitive ratios $1/(n-1)$ for $n\geq 3$~\cite{journals/tcs/TanW07} and $2/3$ for $n=2$~\cite{journals/orl/KellererKST97} have been proved.
	In this section, we focus on the case when agents have general additive valuation functions.
	We first show that no online algorithm can guarantee a competitive ratio larger than $0$, even for $n=3$ agents.
    Then we propose our $0.5$-competitive algorithm for the two-agent case.
	
	Note that for the allocation of goods, we have $\MMS_i \leq (1/n)\cdot v_i(M) = 1$.
	Therefore as long as an agent receives a bundle with $v_i(X_i) \geq \gamma$, the allocation must be at least $\gamma$-MMS to her.
	
	\subsection{Upper Bound of Approximation Ratio}
	
	We show in this subsection that when there are at least $3$ agents, the problem of approximating MMS allocations online does not admit any competitive algorithm.
	
	\begin{theorem} \label{theorem:goods-n-agents}
		No online algorithm has a competitive ratio strictly larger than $0$ for approximating MMS allocations for goods, even when $n=3$.
	\end{theorem}
	\begin{proof}
		We first consider the case of $n=3$ agents and provide a collection of instances showing that no online algorithm can guarantee a competitive ratio strictly larger than $0$ on these instances.
		The case when $n\geq 4$ can be proved in a very similar way, and we defer the proof to \cref{sec:hardness-goods->=4}.

		For the sake of contradiction, suppose there exists a $\gamma$-competitive algorithm for approximating MMS allocation for $n=3$ agents, where $\gamma \in (0,1]$.
		Let $r > 1/\gamma$ be a sufficiently large integer and $\epsilon > 0$ be sufficiently small such that $r^3\epsilon < \gamma$.
		We construct a collection of instances and show that in the allocation returned by the algorithm for at least one of these instances, at least one agent $i\in N$ is allocated a bundle $X_i$ with $v_i(X_i) < (1/r)\cdot \MMS_i$, which is a contradiction.
		Note that since the allocation is deterministic, we can construct an instance gradually depending on how the previous items are allocated.
		
		To begin with, let the first item be $e_1$ with $v_1(e_1) = v_2(e_1) = v_3(e_1) = \epsilon$ and assume w.l.o.g. that agent $1$ receives it.
		Then let the second item be $e_2$ with $v_1(e_2) = r^2\epsilon,v_2(e_2) = v_3(e_2) = \epsilon$.
		Since the online algorithm has competitive ratio $\gamma > 1/r$, item $e_2$ can not be assigned to agent $1$.
		This is because otherwise for the instance with only three items, where $v_1(e_3) = 3-\epsilon-r^2\epsilon$ and $v_2(e_3) = v_3(e_3) = 3-2\epsilon$, there must exists an agent (in $\{2,3\}$) that receives no item, which leads to a $0$-MMS allocation.
		Since $v_2(e_1) = v_2(e_2) = v_3(e_1) = v_3(e_2)$, we can assume w.l.o.g. that agent $2$ receives item $e_2$.
		Let $e_3$ be such that $v_1(e_3) = r\epsilon, v_2(e_3) =r^2\epsilon, v_3(e_3) = \epsilon$.

		We first show that $e_3$ can not be assigned to agent $2$.
		Assume otherwise, i.e., $e_3\in X_2$.
		Then for the instance with four items, where the last item $e_4$ has 
		\begin{align*}
			v_1(e_4) = 3-(\epsilon+r\epsilon+r^2\epsilon),\quad v_2(e_4) = 3-(2\epsilon+r^2\epsilon), \quad v_3(e_4) = 3-3\epsilon,
		\end{align*}
		either $X_1 = \{e_1\}$ or $X_3 = \emptyset$ in the final allocation.
		Since $\MMS_1 = \epsilon+r\epsilon$ and $\MMS_3 >0$, in both cases the competitive ratio is strictly smaller than $\gamma$.
		Hence every $\gamma$-competitive algorithm must allocate item $e_3$ to either agent $1$ or $3$.
		Depending on which agent receives item $e_3$, we construct two different instances.
		
		For the case when agent $1$ receives item $e_3$, we construct the instance with five items as in Table~\ref{tab:Case1}.
		
		\begin{table}[htb]
            \centering
			\begin{tabular}{ c|c|c|c|c|c } 
				$\qquad$ & $e_1$ & $e_2$ & $e_3$ & $e_4$ & $e_5$ \\
				\hline
				$\mathbf{1}$ & $\boxed{\epsilon}$ & $r^2\epsilon$ & $\boxed{r\epsilon}$ & $r^3\epsilon$ & $3-(r^3\epsilon+r^2\epsilon+r\epsilon+\epsilon)$\\ 
				\hline
				$\mathbf{2}$ & $\epsilon$ & $\boxed{\epsilon}$ & $r^2\epsilon$ & $r\epsilon$ & $3-(r^2\epsilon+r\epsilon+2\epsilon)$ \\
				\hline
				$\mathbf{3}$ & $\epsilon$ & $\epsilon$ & $\epsilon$ & $\epsilon$ & $3-4\epsilon$
			\end{tabular}
        	\caption{Instance with five items when item $e_3$ is assigned to agent $1$.}\label{tab:Case1}
		\end{table}
		
		Let $X_i'$ be the bundle agent $i$ holds at the moment, for each $i\in N$. 
		We have
		\begin{align*}
			\MMS_1 = r^2\epsilon+r\epsilon+\epsilon, \quad \MMS_2 = r\epsilon+2\epsilon, \quad \MMS_3 = 2\epsilon \quad \text{and} \\
			v_1(X_1') = v_1(e_1+e_3) = r\epsilon+\epsilon, \quad v_2(X_2') = v_2(e_2) = \epsilon, \quad v_3(X_3') = 0.
		\end{align*}
		Note that for all $i\in N$ we have $v_i(X_i') < 1/r\cdot \MMS_i$.
		Observe that there must exist at least one agent $i \in N$ that does not receive any item in $\{e_4,e_5\}$ in the final allocation.
		In other words, we have $X_i = X'_i$, which leads to a contradiction that the algorithm computes $\gamma$-MMS allocations.
		
		Next we consider the case when agent $3$ receives item $e_3$.
		Let the next item $e_4$ be such that $v_1(e_4) = \epsilon, v_2(e_4) = r\epsilon, v_3(e_4) = r^2\epsilon$.
		We argue that the algorithm can not allocate item $e_4$ to agent $3$.
		Consider the following instance (as shown in Table~\ref{tab:Case2a}).
		
		\begin{table}[htb]
		    \centering
			\begin{tabular}{ c|c|c|c|c|c } 
				& $e_1$ & $e_2$ & $e_3$ & $e_4$ & $e_5$ \\
				\hline
				$\mathbf{1}$ & $\boxed{\epsilon}$ & $r^2\epsilon$ & $r\epsilon$ & $\epsilon$ &  $3-(r^2\epsilon+r\epsilon+2\epsilon)$ \\ 
				\hline
				$\mathbf{2}$ & $\epsilon$ & $\boxed{\epsilon}$ & $r^2\epsilon$ & $r\epsilon$ &  $3-(r^2\epsilon+r\epsilon+2\epsilon)$\\
				\hline
				$\mathbf{3}$ & $\epsilon$ & $\epsilon$ & $\boxed{\epsilon}$ & $\boxed{r^2\epsilon}$ & $3-(r^2\epsilon+3\epsilon)$ 
			\end{tabular}
		    \caption{Instance showing why item $e_4$ can not be assigned to agent $3$.}\label{tab:Case2a}
		\end{table}
		
		Note that for this instance we have $\MMS_1 = \MMS_2 = r\epsilon+2\epsilon$ while at the moment we have $v_1(X_1') = v_2(X_2') = \epsilon$.
		Therefore the agent $i\in \{1,2\}$ that does not receive item $e_5$ in the final allocation is not $\gamma$-MMS.
		Therefore we know that the algorithm must allocate item $e_4$ to agent $1$ or $2$, for which case we construct the following instance with six items (as shown in Table~\ref{tab:Case2b}).
		
		\begin{table}[htb]
			\centering
			\begin{tabular}{ c|c|c|c|c|c|c } 
				& $e_1$ & $e_2$ & $e_3$ & $e_4$ & $e_5$ & $e_6$ \\
				\hline
				$\mathbf{1}$ & $\boxed{\epsilon}$ & $r^2\epsilon$ & $r\epsilon$ & $\epsilon$ & $r^2\epsilon$ & $3-(2r^2\epsilon+r\epsilon+2\epsilon)$ \\ 
				\hline
				$\mathbf{2}$ & $\epsilon$ & $\boxed{\epsilon}$ & $r^2\epsilon$ & $r\epsilon$ & $r^3\epsilon$ & $3-(r^3\epsilon+r^2\epsilon+r\epsilon+2\epsilon)$\\
				\hline
				$\mathbf{3}$ & $\epsilon$ & $\epsilon$ & $\boxed{\epsilon}$ & $r^2\epsilon$ & $r^2\epsilon$ & $3-(2r^2\epsilon+3\epsilon)$ 
			\end{tabular}
			\caption{Suppose that item $e_4$ is allocated to one of the agents in $\{1,2\}$.}\label{tab:Case2b}
		\end{table}
		
		Note that for this instance we have
		\begin{equation*}
		    \MMS_1 = r^2\epsilon+2\epsilon, \quad 
		    \MMS_2 = r^2\epsilon+r\epsilon+2\epsilon, \quad
		    \MMS_3 = r^2\epsilon+\epsilon.
		\end{equation*}
		
		On the other hand, (no matter which agent receives item $e_4$) we have
		\begin{equation*}
		    v_1(X'_1) \leq 2\epsilon, \quad 
		    v_2(X'_2) \leq r\epsilon + \epsilon, \quad
		    v_3(X'_3) = \epsilon.
		\end{equation*}
		Since only items $\{e_5, e_6\}$ are not allocated, for the agent $i\in N$ that does not receive any item in $\{e_5, e_6\}$ in the final allocation, the allocation is not $\gamma$-MMS to her.

		In summary, no online algorithm is $\gamma$-competitive for three agents.
		Extending this result to $n\geq 4$ agents is almost straightforward, and we defer the proof to the appendix.
	\end{proof}
	
	\subsection{Two Agents} \label{ssec:good-2-agents}
	
	In this section, we consider the case with two agents and show that $0.5$-competitive algorithm exists, and is indeed optimal.
	A natural algorithmic idea for two agents is to allocate each item to the agent who values the items more (in a greedy manner) and stop allocating any further item to an agent once its total value exceeds $0.5$.
    Unfortunately, via the following instance (see Table~\ref{tab:Greedy-Hardness}) we show that the allocation returned by this greedy algorithm can be arbitrarily bad.
    
    \begin{table}[htb]
		\centering
			\begin{tabular}{ c|c|c|c } 
				$\qquad$ & $e_1$ & $e_2$ & $e_3$  \\
				\hline
				$\mathbf{1}$ & $0.5 - 2\epsilon$ & $1.5 - \epsilon$ & $\boxed{3\epsilon}$ \\ 
				\hline
				$\mathbf{2}$ & $\boxed{0.5 - \epsilon}$ & $\boxed{1.5}$ & $\epsilon$
			\end{tabular}
			\caption{Hard instance for the greedy algorithm, where $\epsilon >0$ is arbitrarily small.}\label{tab:Greedy-Hardness}
	\end{table}
	
	In this instance, both items $e_1, e_2$ are allocated to agent $2$ since she values them more, and her value does not exceed $0.5$ at the moment of allocation.
    However, the allocation is far from being MMS fair to agent $1$, since we have $\MMS_1 = 0.5+\epsilon$ while $v_1(X_1) = 3\epsilon$.
    The algorithm fails because when there exists an item with very large value, e.g., item $e_2$ in the instance, the greedy algorithm may allocate this item to the agent with value close to $0.5$, which results in a significantly unfair allocation.
    To fix this issue, we propose the following algorithm that equips the greedy algorithm with a special handling of large items.
    
	\paragraph{Algorithm for Two Agents.}
	We call an item $e$ \emph{large} to agent $i$ if $v_i(e) \geq 0.5$.
	For each online item $e$, if it is large to both agents, we assign $e$ to the agent with smaller $v_i(X_i)$ at the moment and allocate all future items to the other agent.
	Otherwise we allocate $e$ to the agent $i$ with larger $v_i(e)$.
	Once we have $v_i(X_i)\geq 1/2$ for some agent $i\in N$, we allocate all future items to the other agent $j\neq i$ (refer to Algorithm~\ref{alg:general-2-agents-goods}).
	Throughout the algorithm (and all later algorithms), we break ties arbitrarily but consistently, e.g., using the id of agents.
	
	\begin{algorithm}[htb]
		\SetKwInOut{KwIn}{Input}\SetKwInOut{KwOut}{Output}
		\SetKw{Break}{Break}
		\caption{Algorithm-for-2-Agents-for-Goods}\label{alg:general-2-agents-goods} 
		Initialize: $X_1, X_2 \gets \emptyset$ and let $A\gets \{1,2\}$ be the active agents\;
		\For{each online item $e \in M$}{
		    \uIf{$|A| = 1$}
		    {
		        $X_i \gets X_i + e$, where $i\in A$\;
		    }
		    \uElseIf{$v_1(e) \geq 1/2$ and $v_2(e) \geq 1/2$}
		   {
			    $i \gets \argmin_{j \in N}\{v_j(X_j)\}$\;
			    $X_{i} \gets X_{i} + e$\;
			    turn agent $i$ into inactive: $A\gets A\setminus\{i\}$\;
		    }
		    \Else
		    {
			    $i \gets \argmax_{j \in N}\{v_j(e)\}$\;
			    $X_{i} \gets X_{i} + e$\;
			    \If{$v_i(X_{i}) \geq 1/2$}{
			        turn agent $i$ into inactive: $A\gets A\setminus\{i\}$\;
			    }
		    }
		}
		\KwOut{$\bX = (X_1, X_2)$}
	\end{algorithm}
	
	\begin{theorem} \label{theorem:goods-2-agents}
		For $n=2$ agents, Algorithm~\ref{alg:general-2-agents-goods} computes an $0.5$-MMS allocation in $O(m)$ time.
	\end{theorem}
	\begin{proof}		
	    Let $\bX = (X_1,X_2)$ be the final allocation.
		We first show that if no item $e \in M$ is large to both agents, then $\bX$ is $1/2$-MMS.
		Since each item $e$ is allocated to the agent with larger $v_i(e)$, we have $v_i(X_i) \geq v_j(X_i)$ for all $i\in N$ and $j \neq i$.
		Therefore in the final allocation, we have
		\begin{equation*}
		    v_1(X_1) + v_2(X_2) \geq v_1(M) = 2.
		\end{equation*}
		Hence at least one of the agents, say agent $i$, will be turned inactive by the algorithm.
		Let $e$ be the last item agent $i$ receives.
		Then we have $v_i(e) \geq v_j(e)$ for $j\neq i$ and $v_i(X_i - e) < 0.5$ by the design of the algorithm.
		Hence we have
		\begin{equation*}
			v_j(X_i) = v_j(X_i - e) + v_j(e) < 1/2 + 1/2 = 1,
		\end{equation*}
		where the inequality holds because $v_j(X_i - e) \leq v_i(X_i - e)$ (since each item is allocated to the agent that values it more), and that item $e$ is not large to agent $j$.
		Hence we have $v_j(X_j) = v_j(M)-v_j(X_i) > 1$.
		Recall that $v_i(X_i) \geq 1/2$.
		Hence the allocation $\bX$ is a $1/2$-MMS.
		
		Next, we argue that if there exists an item $e\in M$ that is large to both agents, then $\bX$ is also $1/2$-MMS.
		Let item $e$ be the first item that is large to both agents, and suppose that it is allocated to agent $i$.
		Then we have $v_i(X_i)\geq v_i(e) \geq 1/2 \geq 1/2 \cdot \MMS_i$. 
		Next we show that $v_j(X_j) \geq 1/2 \cdot \MMS_j$ for the other agent $j\neq i$.
		Let $X'_1$ (resp. $X'_2$) be the bundle agent $1$ (resp. $2$) holds before item $e$ is allocated.
		By the design of the algorithm we have $v_i(X'_i) \leq v_j(X'_j)$.
		Since all items that arrive before $e$ are allocated greedily, we have $v_j(X'_j) \geq v_i(X'_i) \geq v_j(X'_i)$.
		Since all items that arrive after $e$ are allocated to agent $j$, we have
		\begin{equation*}
		    v_j(X_j) \geq \frac{1}{2}\cdot (v_j(X_j) + v_j(X'_i)) = \frac{1}{2} \cdot v_j(M - e) \geq\frac{1}{2} \cdot \MMS_j,
		\end{equation*}
		where last inequality holds because in any allocation the bundle not containing item $e$ has value at most $v_j(M - e)$, which implies $\MMS_j \leq v_j(M - e)$.
		Since the allocation of every item $e\in M$ takes $O(1)$ time, the algorithm runs in $O(m)$ time.
	\end{proof}
	
	Next, we show that no online algorithm can achieve a competitive ratio strictly better than $0.5$.
	Therefore, our $0.5$-competitive algorithm is optimal.	

	\begin{theorem} \label{thm:0.5for-2agents}
		For $n=2$ agents, no online algorithm has a competitive ratio strictly larger than $0.5$.
	\end{theorem}

	\begin{proof}
	    Assume the contrary that there exists an online algorithm with competitive ratio $0.5+\delta$ for some constant $\delta > 0$.
	    We first prove that this algorithm must maintain the following property throughout the execution of the whole algorithm, which (roughly speaking) enforces the algorithm to maintain an ``uneven'' allocation at all times.
	    
    	\begin{claim}\label{claim:property-2-agents-goods}
    		Throughout the execution of a $(0.5+\delta)$-competitive algorithm, if $v_1(X_1+X_2) \leq 1$ and $v_2(X_1+X_2) \leq 1$, then there must exists $i\in N$ such that $v_i(X_i) > (1+\delta)\cdot v_i(X_j)$.
    	\end{claim}
    	\begin{proof}
    		Suppose otherwise, e.g., $v_i(X_i) \leq (1+\delta)\cdot v_i(X_j)$ for both $i\in N$ at some moment. 
    		Then consider the instance with items $M = X_1\cup X_2 \cup \{e\}$, where the last item $e$ has $v_i(e) = 2 - v_i(X_1+X_2)$ for both $i\in N$.
    		Since $v_1(X_1+X_2) \leq 1$ and $v_2(X_1+X_2) \leq 1$, we have $v_i(e) \geq 1$ for both $i\in N$.
    		Therefore, the agent who does not receive item $e$ in the final allocation has value $v_i(X_i)$, while her MMS is $v_i(X_i+X_j)$.
    		Note that
    		\begin{equation*}
    		    v_i(X_i) \leq \cfrac{1+\delta}{2+\delta} \cdot v_i(X_i+X_j) \leq \cfrac{1+\delta}{2+\delta}\cdot \MMS_i < (0.5+\delta) \cdot \MMS_i,
            \end{equation*}
            which contradicts the fact that the algorithm computes $(0.5+\delta)$-MMS allocations.
        \end{proof}
    
		Following Claim~\ref{claim:property-2-agents-goods} we construct an instance (see Table~\ref{tab:Two-Goods}), and show that the algorithm must allocate all items $e_1,\ldots,e_k$ to agent $1$ (we can assume w.l.o.g. that agent $1$ receives the first item).
		
    	\begin{table}[htbp]
    		\small
    		\centering
    		\begin{tabular}{c|c|c|c|c|c|c|c}
    			 & $e_1$ & $e_2$  & $\cdots$ & $e_l$ & $\cdots$ & $e_k$ & $e_{k+1}$ \\ \hline
    			\textbf{1} & $\epsilon$ & $\cfrac{\epsilon}{1 + \delta}$  & $\cdots$ & 
    			$\cfrac{(2+\delta)^{l-2}\cdot \epsilon}{(1+\delta)^{l-1}}$ & $\cdots$ &
    			$\cfrac{(2+\delta)^{k-2}\cdot \epsilon}{(1+\delta)^{k-1}}$ & $ > 1.9$
    			\\ \hline
    			\textbf{2} & $\epsilon$ & $(1+\delta)\! \cdot\! \epsilon$ & $\cdots$ & $(2+\delta)^{l-2}\!\cdot\!(1+\delta)\!\cdot\! \epsilon$ & $\cdots$ &
    			$(2+\delta)^{l-2}\!\cdot\!(1+\delta)\!\cdot\! \epsilon$ & $ < 0.1$
    		\end{tabular}
        	\caption{Hard instance for allocation of goods for two agents, where $\epsilon>0$ is arbitrarily small.}\label{tab:Two-Goods}
    	\end{table}

		Specifically, the first item $e_1$ has $v_1(e_1) = v_2(e_1) = \epsilon$ and we can assume w.l.o.g that agent $1$ receives this item.
		Then we construct items $e_2, \cdots, e_k$ such that for all $1< i\leq k$, we have
		\begin{align*}
		    v_1(e_i) = \frac{(2+\delta)^{i-2}\cdot \epsilon}{(1+\delta)^{i-1}} \quad \text{and} \quad
		    v_2(e_i) =
		    \begin{cases}
		    (2+\delta)^{i-2}\cdot(1+\delta)\cdot \epsilon, \quad \text{ if } i\leq l\\
            (2+\delta)^{l-2}\cdot(1+\delta)\cdot \epsilon, \quad \text{ if } i > l.
		    \end{cases}
		\end{align*}
		Observe that for all $1<i\leq k$, we have
		\begin{align}
		    \sum_{j=1}^{i-1}v_1(e_j) & = \epsilon + \frac{\epsilon}{1+\delta}\cdot \sum_{j=2}^{i-1} \left(\frac{2+\delta}{1+\delta}\right)^{j-2} 
		    = \epsilon + \frac{\epsilon}{1+\delta} \cdot (1+\delta)\cdot \left(\left(\frac{2+\delta}{1+\delta}\right)^{i-2} - 1 \right) \nonumber \\
		    & = \epsilon\cdot \left(\frac{2+\delta}{1+\delta}\right)^{i-2} = (1+\delta)\cdot v_1(e_i). \label{eqn:sum-of-v1(e)}
		\end{align}
		For all $1<i\leq l$, we have
		\begin{align}
		    \sum_{j=1}^{i-1}v_2(e_j) & = \epsilon + \epsilon\cdot (1+\delta)\cdot \sum_{j=2}^{i-1} (2+\delta)^{j-2} 
		    = \epsilon + \epsilon\cdot (1+\delta)\cdot \frac{1}{1+\delta}\cdot \left( (2+\delta)^{i-2} - 1 \right) \nonumber \\
		    & = \epsilon\cdot (2+\delta)^{i-2} = \frac{v_2(e_i)}{1+\delta}.\label{eqn:sum-of-v2(e)}
		\end{align}
		In other words, item $e_i$ has value roughly the same as all items before $i$ combined, for both agents.
		However, the value grows slightly faster in $v_2$ than in $v_1$.
		This is also reflected by the definition of $v_1(e_i)$ and $v_2(e_i)$: observe that for all $1<i\leq l$ we have $\frac{v_2(e_i)}{v_1(e_i)} = (1+\delta)^i$, i.e., the difference in value under the two valuation functions grows exponentially in $i$.
		Note that all items $e_{l}, e_{l+1},\ldots, e_k$ have the same value $(2+\delta)^{l-2} \cdot (1+\delta) \cdot \epsilon$ to agent $2$.
		By picking $\epsilon = \frac{0.1}{(1+\delta)(2+\delta)^{l-2}}$ we can ensure that 
		\begin{equation*}
		    v_2(e_l) = v_2(e_{l+1}) = \cdots = v_2(e_k) = (2+\delta)^{l-2} \cdot (1+\delta) \cdot \epsilon = 0.1.
		\end{equation*}
		Finally, let $e_{k+1}$ be the last item with $v_1(X_1+X_2+e_{k+1}) = v_2(X_1+X_2+e_{k+1}) = 2$.
        By setting $k = l + 18$, we have (where the third equality follows from~\eqref{eqn:sum-of-v2(e)} and $v_2(e_l)=0.1$)
        \begin{equation*}
            v_2(e_{k+1}) = 2 - \sum_{i=1}^k v_2(e_i) = 2 - \sum_{i=1}^{l-1} v_2(e_i) - (k-l+1)\cdot v_2(e_l) = 2 - \frac{0.1}{1+\delta} - 1.9 = \frac{0.1\cdot \delta}{1+\delta}.
        \end{equation*}
		By setting $l$ to be sufficiently large (which also defines $\epsilon$ and $k$), we can ensure that
        \begin{equation*}
            v_1(e_{k+1}) = 2 - \sum_{i=1}^k v_1(e_i) = 2 - \epsilon\cdot \left(\frac{2+\delta}{1+\delta}\right)^{k-1} = 2 - \frac{0.1\cdot (2+\delta)^{19}}{(1+\delta)^{l+18}} > 1.9.
        \end{equation*}
		Next, we argue that the algorithm must allocate the first $k$ items to agent $1$.
		\begin{claim}\label{claim:all-first-k-to-1}
			The algorithm with competitive ratio $0.5+\delta$ must assign all items $e_2, \cdots, e_k$ to agent $1$. 
		\end{claim}
    	\begin{proof}
    	    We first prove that all items $\{ e_1,\ldots, e_l \}$ must be allocated to agent $1$ by showing that for all $i\in \{2,\ldots,l\}$, if agent $1$ receives all the items $e_1,\cdots,e_{i-1}$, then $e_i$ must also be allocated to agent $1$, using Lemma~\ref{claim:property-2-agents-goods}.
    		Recall from~\eqref{eqn:sum-of-v1(e)} and~\eqref{eqn:sum-of-v2(e)} that
    		\begin{equation*}
    		    \sum_{j=1}^{i-1} v_1(e_j) = (1+\delta)\cdot v_1(e_i) \quad \text{ and } \quad \sum_{j=1}^{i-1} v_2(e_j) = \frac{v_2(e_i)}{1+\delta}.
    		\end{equation*}
    		Suppose item $e_i$ is allocated to agent $2$, then after the allocation we have $X_1 = \{ e_1,\ldots,e_{i-1} \}$ and $X_2 = \{e_i\}$, which gives
    		\begin{equation*}
    		    v_1(X_1) = (1+\delta)\cdot v_1(X_2) \quad \text{and} \quad v_2(X_2) = \frac{v_2(X_1)}{1+\delta}.
    		\end{equation*}
    		
    		Moreover, we have $v_i(X_1\cup X_2) \leq 1$ for both $i\in \{1,2\}$.
    		By Lemma~\ref{claim:property-2-agents-goods}, we have a contradiction.
    		
    		Next we show that for all $i\in \{l+1,\ldots,k\}$, if agent $1$ receives all the items $e_1,\cdots, e_{i-1}$, then item $e_i$ should also be assigned to her.
    		Assume otherwise and let $e_i$, where $i\in \{l+1,\ldots,k\}$ be the first item agent $2$ receives.
    		Recall that $v_2(e_l) = \cdots = v_2(e_k) = 0.1$.
    		Hence after allocating item $e_i$ we have $v_2(X_2) = 0.1$ and $v_1(X_1) = (1+\delta) \cdot v_1(X_2)$.
    		Then consider the instance with $i+1$ items, where the last item $e_{i+1}$ has values $v_1(e_{i+1}) = 2 - v_1(X_1\cup X_2)$ and $v_2(e_{i+1}) = 2 - v_2(X_1\cup X_2)$.
    		Note that for this instance we have
    		\begin{equation*}
    		    \MMS_1 = v_1(X_1\cup X_2)\quad \text{ and } \quad \MMS_2 \geq \sum_{j=1}^{l+1} v_2(e_j) = \frac{3+2\delta}{1+\delta}\cdot 0.1
    		\end{equation*}
    		Therefore, the agent who does not receive $e_{i+1}$ in the final allocation will have value strictly less than $(0.5+\delta)$ times her MMS, which is also a contradiction.
    	\end{proof}
	
		Given Claim~\ref{claim:all-first-k-to-1}, we know that when item $e_{k+1}$ arrives, agent $2$ is not allocated any item, which leads to $v_2(X_2) < 0.1$ in the final allocation. 
		Since $\MMS_2 = 1$, the allocation is obviously not $(0.5+\delta)$-MMS to agent $2$, which is also a contradiction.
	\end{proof}
	
\subsection{Monotone Instances}\label{ssec:mono-goods}
	
	In this section, we consider monotone instances for the allocation of goods, for which we propose a $0.5$-competitive algorithm.
	Suppose the set of online items are indexed by $\{e_1,e_2,\ldots,e_m\}$ following their arrival order.
	
	\begin{definition}[Monotone Instances for Goods]
		We call an online instance \emph{monotone} if for all agent $i\in N$, we have
		\begin{equation*}
			v_i(e_1)\geq v_i(e_2)\geq \cdots \geq v_i(e_m).
		\end{equation*}
	\end{definition}
	
	As discussed in Section~\ref{sec:introduction}, the online fair allocation problem is difficult due to (1) the irrevocable decision making and (2) the arbitrary arrival order of items.
	The monotone instances somehow remove the second difficulty in this problem.
	By providing a $0.5$-competitive algorithm for this case we show that the problem becomes much easier with the second difficulty removed.
	
	\begin{theorem}{\label{theorem:monotone-for-goods}}
		There exists a polynomial time $0.5$-competitive algorithm for monotone instances.
	\end{theorem}
	
	The following useful property regarding MMS for goods are commonly observed by existing works on the approximation of MMS allocations~\cite{journals/jacm/KurokawaPW18,conf/soda/GargMT19,journals/ai/GargT21}.
	For completeness, we also present a short proof.
	
	\begin{lemma}{\label{lemma:MMS'-greaterthan-MMS}}
		For all $N'\subseteq N$ and $M'\subseteq M$, let $\MMS_i(N', M')$ be the maximin share of agent $i\in N'$ (under valuation function $v_i$) on the instance with agents $N'$ and items $M'$. If $|N\setminus N'| = |M\setminus M'|$, then we have $\MMS_i(N', M') \geq \MMS_i(N,M)$.
	\end{lemma}
    \begin{proof}
        Let $\bX = \{X_1,\cdots,X_n\}$ is the partition of items $M$ under the valuation function $v_i$ that achieves the maximin value $\MMS_i$. 
        Then we have $v_i(X_j) \geq \MMS_i$ for any $j\in N$.
        Let $k = |N\setminus N'| = |M\setminus M'|$, we can find at most $k$ bundles in $\bX$ (say $X_1, X_2, \cdots X_k$) such that any item in $M\setminus M'$ belongs to some bundle in $\{X_1, \cdots, X_k\}$, i.e. $M\setminus M' \subseteq X_1\cup X_2\cup\cdots\cup X_k$.
        Let $M^- = M\setminus \{X_1\cup X_2\cup\cdots\cup X_k\}$.
        Then we know that $\{X_{k+1},\cdots, X_n\}$ is a $(n-k)$-partition of items $M^-$ while $M^-$ is a subset of $M'$.
        Hence we have $\MMS_i(N', M') \geq \MMS_i(N', M^-)\geq \min_{j > k}v_i(X_j) \geq \MMS_i$.
    \end{proof}

	\paragraph{Algorithm for Monotone Instances.}
	We maintain a set of active agents $A\subseteq N$, and we only allocate online items to active agents.
	Initially, all agents are active, and we inactivate an agent $i$ once we can ensure that she receives a bundle of value at least $0.5\cdot \MMS_i$.
	In general, we follow the principle of distinguishing large and small items, as in Section~\ref{ssec:good-2-agents}.
	However, we dynamically update the threshold for defining large items throughout the execution of the algorithm, with the guarantee that once an agent receives an item that is large to her, the allocation will be $0.5\cdot \MMS_i$ to her, no matter what items would arrive in the future.
	Our algorithm has two phases.
	In phase-$1$ we allocate the large items, each of which to a unique agent that will be inactivated immediately after the allocation.
	When the online item is small to all agents, our algorithm enters phase-$2$, during which the items are allocated greedily to the remaining agents (see Algorithm~\ref{alg:monotone-goods}).
	
	\begin{algorithm}[htb]
		\SetKwInOut{KwIn}{Input}\SetKwInOut{KwOut}{Output}
		\SetKw{Break}{Break}
		\caption{Algorithm-for-Monotone-Instances-for-Goods}\label{alg:monotone-goods} 
		Initialize: $\text{inPhase1}\gets \textbf{true}$, $A\gets N$, $L\gets \emptyset$ and for any $i\in N$, $X_i\gets \emptyset$ \;
		\For{each online item $e\in M$}{
			\uIf{$\text{inPhase1} = \textbf{true}$}{
				\uIf{there exists $i\in A$ with $v_i(e) \geq \frac{v_i(M\setminus L)}{2|A|}$}{
					$X_{i}\gets X_{i} + e$, $A\gets A\setminus\{ i \}$\;
					update the set of arrived large items: $L \gets L + e$ \;
				}
				\Else{
					$\text{inPhase1}\gets \textbf{false}$ \;
					define $\beta_i \gets  \frac{v_i(M\setminus L)}{2|A|}$ for all $i\in N$ \;
				}
			}
			\uElseIf{$|A| \geq 2$}{
				$i\gets \argmax_{j\in A}\{\frac{v_j(e)}{\beta_j}\}$\;
				$X_{i}\gets X_{i}+e$\;
				\If{ $v_i(X_i) \geq \beta_i$ }{
					$A\gets A\setminus \{ i \}$ \;
				}
			}
			\Else{
				assign $e$ to the only active agent \;
			}
			
		}
		\KwOut{$\mathbf{X} = (X_1, X_2,\cdots, X_n)$}
	\end{algorithm}

	\begin{proofof}{Theorem~\ref{theorem:monotone-for-goods}}
		We show that Algorithm~\ref{alg:monotone-goods} returns a $0.5$-MMS allocation.
		Observe that in phase-$1$ we inactivate an agent once it receives an item (which is large).
		Therefore we always have $|A| = n - |L|$.
		Then by Lemma~\ref{lemma:MMS'-greaterthan-MMS}, we have $\MMS_i(A, M\setminus L) \geq \MMS_i$ for all agent $i\in N$.
		If agent $i$ is inactivated in phase-$1$, then we have
		\begin{equation*}
		v_i(X_i) \geq \frac{v_i(M\setminus L)}{2|A|}\geq \frac{1}{2}\cdot \MMS_i(A, M\setminus L) \geq \frac{1}{2}\cdot \MMS_i.
		\end{equation*}
		
		Now we consider the moment when the algorithm enters phase-$2$ and let $A^*$ be the agents that are still active.
		We define $\beta_i= \frac{v_i(M\setminus L)}{2|A^*|}$ to be half the average value when dividing items in $M\setminus L$ to agents in $A^*$ under $v_i$, and we have $\beta_i \geq \frac{1}{2}\cdot \MMS_i$.
		Note that the algorithm enters phase-$2$ when the online item $e$ has value $v_i(e) \leq \beta_i$ to all agents $i\in A^*$.
		Since the instance is monotone, we know that in phase-$2$, for all $i\in A^*$, every online item $e$ has value $v_i(e) \leq \beta_i$.
		Recall that our algorithm inactivates an agent in phase-$2$ if $v_i(X_i) \geq \beta_i \geq \frac{1}{2}\cdot \MMS_i$, unless she is the only active agent.
		Therefore the allocation is $0.5$-MMS to these inactivated agents.
		Moreover, we show that at any point in time during the execution of the algorithm, if $j\in A^*$ is active and $i\in A^*$ is inactivated, then we have $v_j(X_i) < 2\beta_j$.
		Let $e$ be the last item allocated to $X_i$.
		Since before allocating item $e$ to $X_i$, agent $i$ is active, we have $v_i(X_i - e) < \beta_i$, which gives $v_i(X_i) < 2\beta_i$ because $v_i(e)\leq \beta_i$.
		By the greedy allocation, for each item $e'\in X_i$ we have $\frac{v_i(e')}{\beta_i} \geq \frac{v_j(e')}{\beta_j}$, which implies
		\begin{equation*}
		v_j(X_i) = \sum_{e' \in X_i} v_j(e') \leq \sum_{e' \in X_i} (\frac{\beta_j}{\beta_i}\cdot v_i(e')) \leq \frac{\beta_j}{\beta_i} \cdot v_i(X_i) < 2\beta_j.
		\end{equation*}
		Following a similar argument we can show that for any two active agents $j$ and $j'$ we have $v_j(X_{j'}) < \beta_j$.
		Therefore, if there exists an agent $i\in A^*$ with $v_i(X_i) < 0.5\cdot \MMS_i \leq \beta_i$ in the final allocation, then for all agent $j\in A^*$ we have $v_i(X_j) < 2\beta_i$, which leads to a contradiction that
		\begin{equation*}
		v_i(M\setminus L) = \sum_{j\in A^*} v_i(X_j) < |A^*|\cdot 2\beta_i = v_i(M\setminus L).
		\end{equation*}
		Hence at the end of the algorithm, we can guarantee that all agent $i\in N$ has $v_i(X_i) \geq 0.5\cdot \MMS_i$.
		Finally, note that the allocation of each item takes $O(n)$ time and the $0.5$-MMS allocation can be computed in $O(n m)$ time.
	\end{proofof}

    \subsection{Instances with Small Items}
    
    In this section, we consider small-item instances for the allocation of goods.
    \begin{definition}[Small-good Instances]
    We call an online instance \emph{small-good} instance if there exists some $\alpha < 1$ such that
        \begin{equation*}
            \forall i\in N, e\in M, v_i(e) \leq \alpha.
        \end{equation*}
    \end{definition}

    Similar to the monotone instances, small items instances somehow remove the difficulty from the arbitrary arrival of items, but with less restriction.
    We show that the problem is much easier to solve under such a condition, by demonstrating the existence of a $(1-\alpha)$-competitive algorithm.
    
    \paragraph{Algorithm for Small-good Instances.}
    Let all agents be active at the beginning and let $A$ be the set of active agents.
    For each online item $e$, we greedily allocate it to the active agent that has maximum value on $e$.
    Once an agent receives a collection of items of value at least $1-\alpha$, we inactivate this agent.
    If at some moment only one agent is active, then all future items will be allocated to this agent (see Algorithm~\ref{alg:small-goods}).
    
	\begin{algorithm}[htb]
		\SetKwInOut{KwIn}{Input}\SetKwInOut{KwOut}{Output}
		\SetKw{Break}{Break}
		\caption{Algorithm-for-Small-Good-Instances}\label{alg:small-goods} 
		Initialize: $A\gets N$ and for any $i\in N$, $X_i\gets \emptyset$\;
		\For{each online item $e\in M$}{
			\uIf{$|A|=1$}{
				$X_i\gets X_i + e$, where $i\in A$\;
				}
            \Else{
                $i\gets \argmax_{j\in A}\{v_j(e)\}$;
                $X_{i}\gets X_{i}+e$\;
				\If{$v_i(X_i) \geq 1 - \alpha$}{
					turn agent $i$ into inactive: $A\gets A\setminus \{i\}$\;
                    }
			}
            }
		\KwOut{$\mathbf{X} = (X_1, X_2,\cdots, X_n)$}
	\end{algorithm}

    \begin{theorem}\label{theorem:small-for-goods}
        For small-good instances, Algorithm~\ref{alg:small-goods} computes a $(1-\alpha)$-MMS allocation in $O(m n)$ time.
    \end{theorem}

    \begin{proof}
        Consider the moment when the last item $e$ arrives.
        Let $A$ be the set of active agents, and let $i\in A$ be an arbitrary agent in $A$.
        For all agent $j\in N\setminus A$ we have
        \begin{equation*}
            v_i(X_j) \leq v_j(X_j) < 1 - \alpha + \alpha = 1,
        \end{equation*}
        where the first inequality is due to the greedy allocation of items and the second inequality follows from the fact that agent $j$ was active before she receives her last item and all items are small to all agents.
        In addition, for all $i'\in A\setminus\{i\}$ (if any), we have
        \begin{equation*}
            v_i(X_{i'})\leq v_{i'}(X_{i'}) < 1-\alpha,
        \end{equation*}
        since both $i$ and $i'$ are active and items are allocated greedily.
        
        Therefore we conclude that $A = \{i\}$, because otherwise we have the contradiction that
        \begin{equation*}
            v_i(M) = v_i(X_i) + \sum_{j\in N\setminus A} v_i(X_j) + \sum_{i'\in A\setminus \{i\}} v_i(X_{i'}) + v_i(e) < n.
        \end{equation*}
        
        Note that for any inactive agent, the allocation is $(1-\alpha)$-MMS.
        Hence it suffices to show that $v_i(X_i+e)\geq 1-\alpha$ for the last active agent $i\in A$, which holds because
        \begin{equation*}
            v_i(X_i+e) =  v_i(M) - \sum_{j\in N\setminus A} v_i(X_j) > n - (n-1) = 1.
        \end{equation*}
        
        Since the allocation of each item takes $O(n)$ time, the algorithm returns a $(1-\alpha)$-MMS allocation in $O(m n)$ time.
    \end{proof}

	\section{Allocation of Chores} \label{sec:chores}
	
	In this section, we consider the allocation of chores.
	Recall that for the allocation of chores, each agent $i\in N$ has a cost $c_i(e)\geq 0$ on item $e\in M$, and the competitive ratios are at least $1$ (thus the smaller the better).
	Also recall that when agents have identical cost functions, optimal competitive ratios $1.585$~\cite{journals/scheduling/KellererKG15} and $4/3$~\cite{journals/orl/KellererKST97} have been proved for $n\geq 3$ and $n=2$, respectively.
	We focus on the case when agents have general additive cost functions.
	In contrast to the allocation of goods, we show that MMS allocation of chores admits constant competitive algorithms even for general number of agents.
	For $n\geq 3$ agents, we propose a $(2-1/n)$-competitive algorithm using a similar idea of greedy allocation as in~\cite{conf/www/LiLW22}; for $n=2$ agents, we improve the competitive ratio to $\sqrt{2}$ and show that the competitive ratio is at least $15/11$ for any online algorithm.
	Note that for the allocation of chores, we have $\MMS_i \geq (1/n)\cdot c_i(M) = 1$.
	Therefore as long as an agent receives a bundle with cost $c_i(X_i) \leq \alpha$, the allocation must be $\alpha$-MMS to her.
	Moreover, we have $\MMS_i \geq \max_{e\in M} \{c_i(e)\}$, because in any allocation the bundle with maximum cost should have cost at least as large as that of the most costly item.
	
	\subsection{Online Approximation Algorithm}
	
	We first consider the general case with $n$ agents and present our algorithm that computes a $(2-1/n)$-MMS allocations in $O(mn)$ time.
	The algorithm follows a similar idea as we have used in the previous section  (refer to Algorithm~\ref{alg:general-n-agents-chores}): (1) each online item is greedily allocated to the active agent that has minimum cost on the item; (2) once an agent receives a collection of items of large cost ($\geq 1-1/n$ in our algorithm), we inactivate this agent.
	Initially, all agents are active, and if at some moment only one agent is active, then all future items will be allocated to this agent.
	
	\begin{algorithm}[htb]
		\SetKwInOut{KwIn}{Input}\SetKwInOut{KwOut}{Output}
		\SetKw{Break}{Break}
		\caption{Algorithm-for-$n$-Agents-for-Chores}\label{alg:general-n-agents-chores} 
		Initialize: $A\gets N$ and for any $i\in N$, $X_i\gets \emptyset$ \;
		\For{each online item $e\in M$}{
		    \uIf{$|A|=1$}{
		        $X_i\gets X_i + e$, where $i\in A$ \;
		    }
			\Else
			{
				$i\gets \argmin_{j\in A}\{c_j(e)\}$\;
				$X_{i}\gets X_{i}+e$\;
				\If{$c_i(X_i) \geq 1 - 1/n$}{
				    turn agent $i$ into inactive: $A\gets A\setminus \{i\}$ \;
				}
			}
		}
		\KwOut{$\mathbf{X} = (X_1, X_2,\cdots, X_n)$}
	\end{algorithm}
	
	\begin{theorem}
		For $n\geq 2$ agents, Algorithm~\ref{alg:general-n-agents-chores} computes a $(2- \frac{1}{n})$-MMS allocation in $O(mn)$ time.
	\end{theorem}
	\begin{proof}
		Note that throughout the execution of the algorithm, if $|A|\geq 2$, then all active agents $i\in A$ has $c_i(X_i) < 1-1/n$.
		Therefore, for any agent $i$ that is inactivated, let $e_i$ be the last item agent $i$ receives, then we have
		\begin{equation*}
			c_i(X_i) < 1-1/n + c_i(e_i) \leq (2-1/n) \cdot \MMS_i,
		\end{equation*}
		where the equality follows from the fact that $\MMS_i\geq \max_{e\in M}\{c_i(e)\}$ and $\MMS_i\geq 1$ for the allocation of chores.
		Hence the final allocation is $(2-1/n)$-MMS to all inactive agents, and it remains to consider the case when $|A| = 1$ at the end of the algorithm.
		
		Let $i$ be the only active agent to which the last item is allocated.
		By the design of the algorithm, for each $j\neq i$ (that is already inactivated), we have $c_j(X_j) \geq 1-1/n$.
		Moreover, by the greedy allocation of the algorithm, for all $e\in X_j$ we have $c_j(e)\leq c_i(e)$ (because both agents $i$ and $j$ are active when $e$ arrives).
		Hence we have $c_i(X_j) \geq c_j(X_j)$, which implies
		\begin{equation*}
			c_i(X_i) = c_i(M) - \sum_{j\neq i} c_i(X_j) 
			\leq n - (n-1)\cdot (1-1/n) = 2 - 1/n \leq (2-1/n)\cdot \MMS_i.
		\end{equation*}
		Therefore, the allocation is also $(2-1/n)$-MMS to agent $i$.
		Since the allocation of each item takes $O(n)$ time, the algorithm returns the $(2-1/n)$-MMS allocation in $O(n m)$ time.
	\end{proof}
	
	\subsection{Two Agents} \label{ssec:chores-2-agents}
	
	The above result gives a $1.5$-competitive algorithm for $n=2$ agents.
	In this section we improve the ratio to $\sqrt{2}\approx 1.414$.
	We complement our algorithmic result with hard instances showing that no online algorithm can do better than $15/11\approx 1.364$-competitive.
	The key observation towards this improvement is to mimic the bin-packing algorithms~\cite{conf/aaai/AzizRSW17,conf/sigecom/HuangL21} for approximating MMS for chores.
    In particular, as long as the costs of the two agents do not differ by a factor larger than $\sqrt{2}$, we deviate from the greedy allocation and treat the item as equally costly to the two agents, and allocate the item to a designated agent (agent $1$ in our algorithm).
    In addition, to ensure a bounded competitive ratio, we dynamically update a lower bound $\alpha_i$ for $\MMS_i$, for both $i\in \{1,2\}$.
    Our algorithm makes sure that each allocation of item does not result in $c_i(X_i) \geq \sqrt{2}\cdot \alpha_i$ for both $i\in \{1,2\}$.
	\paragraph{A $\sqrt{2}$-Competitive Algorithm for Two Agents.}
	Throughout the execution of the algorithm, we maintain that $\alpha_i = \max\{1, \max_{e: \text{arrived}} \{c_i(e)\} \}$.
	As argued before, we have $\MMS_i \geq \alpha_i$.
    For each online item $e$, we first identify the agents $A = \{i\in N: c_i(X_i + e) \leq \sqrt{2} \cdot \alpha_i\}$ that can receive the item $e$ without violating the competitive ratio $\sqrt{2}$.
    If $|A| = 1$ then we allocate $e$ to the only agent in $A$; otherwise if $c_1(e) \leq \sqrt{2}\cdot c_2(e)$, we allocate item $e$ to agent $1$; otherwise we allocate item $e$ to agent $2$ (refer to Algorithm~\ref{alg:general-2-agents}).
	\begin{algorithm}[htb]
		\SetKwInOut{KwIn}{Input}\SetKwInOut{KwOut}{Output}
		\SetKw{Break}{Break}
		\caption{Algorithm-for-Two-Agents-for-Chores}\label{alg:general-2-agents} 
		Initialize: $X_1 \gets \emptyset$, $X_2 \gets \emptyset$, $\alpha_1 \gets 1$ and $\alpha_2 \gets 1$ \;
		\For{each online item $e\in M$}{
			update $\alpha_1 \gets \max\{ \alpha_1, c_1(e) \}$ and $\alpha_2 \gets \max\{ \alpha_2, c_2(e) \}$ \;
			$A \gets \{i\in N: c_i(X_i + e) \leq \sqrt{2} \cdot \alpha_i\}$ \;
			\uIf{$|A| = 1$}{
				$X_i \gets X_i + e$, where $i\in A$ \;
			}
			\uElseIf{$c_1(e) \leq \sqrt{2}\cdot c_2(e)$}{
				$X_1 \gets X_1 + e$\;
			}
			\Else{
				$X_2 \gets X_2 + e$\;
			}
		}
		\KwOut{$\mathbf{X} = (X_1, X_2)$}
	\end{algorithm}
	\begin{theorem}
		For $n = 2$ agents, Algorithm~\ref{alg:general-2-agents} computes a $\sqrt{2}$-MMS allocation in $O(m)$ time.
	\end{theorem}
	\begin{proof}
		Since we always have $\MMS_i \geq \alpha_i$, to show that the returned allocation is $\sqrt{2}$-MMS, it suffices to show that when each item $e$ arrives, the set $A$ is not empty.
		Because when $A\neq \emptyset$, our algorithm will allocate $e$ to some agent $i\in A$, which ensures that $c_i(X_i + e)\leq \sqrt{2}\cdot \alpha_i \leq \sqrt{2}\cdot \MMS_i$.
		For the sake of contradiction, we assume that when some online item $e^*$ arrives we have $A = \emptyset$.
		That is, we have $c_i(X_i + e^*) > \sqrt{2}\cdot \alpha_i$ for both $i\in \{1,2\}$, where $X_i$ is the bundle agent $i$ holds when $e$ arrives.
		We claim that under this situation, we have $c_1(e') > \sqrt{2}\cdot c_2(e')$ for all item $e'\in X_2$.
 		\begin{claim}\label{claim:morethansqrt2}
 			For all $e'\in X_2$, $c_1(e') > \sqrt{2}\cdot c_2(e')$.
 		\end{claim}
     	\begin{proof}
		Suppose otherwise and let $e_1$ be the first item assigned to $X_2$ with $c_1(e_1) \leq \sqrt{2}\cdot c_2(e_1)$.
		Let $X'_1$ be the bundle agent $1$ holds when $e_1$ arrives.
		By the design of the algorithm we must have $c_1(X'_1 + e_1) > \sqrt{2}\cdot \alpha_1 \geq \sqrt{2}$, because otherwise $e_1$ will be allocated to agent $1$.
		Recall that we assumed $c_1(X_1 + e^*) > \sqrt{2}\cdot \alpha \geq \sqrt{2}$ for some item $e^*$ that is not allocated in the final allocation, which implies
		\begin{equation*}
		    c_1(X_2) \leq c_1(M) - c_1(X_1 + e^*) < 2 - \sqrt{2}.
		\end{equation*}
		
		Therefore we have $c_1(e_1) \leq c_1(X_2) < 2-\sqrt{2}$ (because $e_1 \in X_2$), which gives
		\begin{equation*}
		    c_1(X'_1) = c_1(X'_1 + e_1) - c_1(e_1) > \sqrt{2} - (2-\sqrt{2}) = 2\sqrt{2} - 2.
		\end{equation*}
        On the other hand, we show that $c_2(X'_1) > 2-\sqrt{2}$.
        Suppose otherwise, then there must exist some $e\in X'_1$ such that $c_1(e) > \sqrt{2}\cdot c_2(e)$, which means that when item $e$ arrives we have $A = \{1\}$.
        Let $X''_1, X''_2$ be the bundles agent $1$ and $2$ hold right before $e$ arrives respectively.
        We have $c_2(X''_2+e) > \sqrt{2}$ since $A = \{1\}$ when item $e$ arrives.
        Since $e \in X'_1$, we have
        \begin{equation*}
            c_2(X''_2) = c_2(X''_2+e) - c_2(e) \geq \sqrt{2}-c_2(X'_1) \geq 2\sqrt{2}-2.
        \end{equation*}
        Recall that $e_1$ is the first item assigned to $X_2$ with $c_1(e_1) \leq \sqrt{2}\cdot c_2(e_1)$, we have $c_1(X''_2) > \sqrt{2} \cdot c_2(X''_2) \geq 4 - 2\sqrt{2}$.
        Then we have $c_1(X_1 + e^*) \leq 2 - c_1(X''_2) < 2\sqrt{2}-2$, which is a contradiction.
        Hence we have $c_2(X'_1) \geq \frac{c_1(X'_1)}{\sqrt{2}} > 2-\sqrt{2}$, which implies 
		$c_2(X_2 + e^*) \leq c_2(M - X'_1) < \sqrt{2}$, and it is a contradiction with our previous assumption that $c_2(X_2 + e^*) > \sqrt{2}\cdot \alpha_2$ for some item $e^*$ (that arrives after $e_1$).
 	    \end{proof}
		Note that $c_1(X_2) \leq c_1(M) - c_1(X_1 + e^*) < 2-\sqrt{2}\cdot \alpha_1 \leq 2-\sqrt{2}$.
		By Claim~\ref{claim:morethansqrt2},
		\begin{equation*}
			c_2(X_2) < \frac{c_1(X_2)}{\sqrt{2}} < \frac{2-\sqrt{2}}{\sqrt{2}} = \sqrt{2}-1.
		\end{equation*}
		Recall that $\alpha_2 \geq \max\{1,c_2(e^*)\}$, we have $c_2(X_2 + e^*)  \leq  \sqrt{2}-1 + c_2(e^*) \leq \sqrt{2}\cdot \alpha_2$,
		which is a contradiction.
		Finally, since the allocation of each item takes $O(1)$ time, the algorithm returns a $\sqrt{2}$-MMS allocation in $O(m)$ time.
	\end{proof}
	
	Next, we present a collection of instances and show that no online algorithm can have a competitive ratio smaller than $15/11$.
	
	\begin{theorem}
		For $n=2$ agents, no online algorithm has a competitive ratio smaller than $15/11$.
	\end{theorem}
	\begin{proof}
		Assume the contrary and suppose there exists an online algorithm with a competitive ratio smaller than $15/11$.
		In the following, we construct some instances and show that the algorithm returns an allocation with an approximation ratio (w.r.t. to MMS) of at least $15/11$ on at least one of the instances, which is a contradiction.
		
		Let the first item be $e_1$ with $c_1(e_1) = c_2(e_1) = 4/11$, and we assume w.l.o.g. that it is allocated to agent $1$.
		Then we construct the second item $e_2$ that values $4/11$ to agent $1$ and $3/11$ to agent $2$.
		We argue that the algorithm with a competitive ratio smaller than $15/11$ can not assign $e_2$ to agent $1$.
		
		\begin{table}[htbp]
		    \centering
		    \begin{tabular}{c|c|c|c|c}
		    	& $e_1$ & $e_2$ & $e_3$ & $e_4$ \\
		    	\hline
		    	\textbf{1} & $\boxed{4/11}$ & $\boxed{4/11}$ & $7/11$ & $7/11$ \\ \hline
		    	\textbf{2} & $4/11$  & $3/11$ & $7/11$   & $8/11$
	    	\end{tabular}
	   	    \caption{Assume item $e_2$ is assigned to agent $1$, we construct the following instance.}\label{tab:Two-Chores-a}
    	\end{table}	
    	
		Assume otherwise, i.e., item $e_2$ is also allocated to agent $1$, and we consider the instance as shown in Table~\ref{tab:Two-Chores-a}.
		Note that for this instance we have $\MMS_1 = \MMS_2 = 1$.
		However, no matter how $e_3$ and $e_4$ are allocated, there must be an agent with a total cost of at least $15/11$, which is a contradiction.
        Therefore the algorithm must allocate $e_2$ to agent $2$.
        Then we construct the following instance with $7$ items (see Table~\ref{tab:Two-Chores}), and show that the algorithm must allocate items $e_3,e_4,e_5$ and $e_6$ to agent $1$.
        
    	\begin{table}[htbp]
    		\centering
	    	\begin{tabular}{c|c|c|c|c|c|c|c}
			    & $e_1$ & $e_2$ & $e_3$ & $e_4$ & $e_5$ & $e_6$ & $e_7$ \\
			    \hline
			    \textbf{1} & $\boxed{4/11}$ & $4/11$         & $\boxed{3/11}$ & $\boxed{3/11}$ & $\boxed{3/11}$ & $\boxed{3/11}$          & $2/11$ \\ \hline
			    \textbf{2} & $4/11$         & $\boxed{3/11}$ & $1/11$       & $1/11$         & $1/11$         & ${1/11}$ & 1 
		    \end{tabular}
            \caption{Assume item $e_2$ is assigned to agent $2$, we construct the following instance.}\label{tab:Two-Chores}
	    \end{table}
	    
	    Assume otherwise and let $e_i$ be the first item in $\{e_3,e_4,e_5,e_6\}$ allocated to agent $2$.
	    Note that right after the allocation we have $c_2(X_2) = 4/11$.
	    Then we consider another instance in which the next item $e_{i+1}$ has
	    \begin{equation*}
	        c_1(e_{i+1}) = 
	        \begin{cases}
	        1, \qquad\ \text{ if } i=3 \\
	        8/11, \quad \text{if } i=4 \\
	        5/11, \quad \text{if } i=5 \\
	        2/11, \quad \text{if } i=6
	        \end{cases}
	        \text{ and }
	        \quad
	        c_2(e_{i+1}) = 1.
	    \end{equation*}
	   
	    It can be verified that in all cases, $\MMS_1 = \MMS_2 = 1$ but whoever receives item $e_{i+1}$ would have cost at least $15/11$, which is a contradiction.
		Hence we have $\{e_3, e_4, e_5,e_6\} \subseteq X_1$, which is also a contradiction because $c_1(X_1) \geq 16/11 > 15/11$ (see Table~\ref{tab:Two-Chores}).
	\end{proof}
	
\subsection{Monotone Instances}\label{ssec:mono-chores}

	In this section, we consider monotone instances for the allocation of chores.
	Suppose the online items are indexed by $\{ e_1, \ldots, e_m \}$ following their arrival order.
	
	\begin{definition}[Monotone Instances for Chores]
		We call an online instance \emph{monotone} if for all agent $i\in N$, we have
		\begin{equation*}
			c_i(e_1)\geq c_i(e_2)\geq \cdots \geq c_i(e_m).
		\end{equation*}
	\end{definition}
	Since the agents agree on the same ordering of items, we can apply the algorithms for approximating MMS for identical ordering (IDO) instances that do not require reallocation of items.
    In particular, we use the following ordinal algorithm called Sesqui-Round Robin proposed by Aziz et al.~\cite{aziz2022approximate}.
    Let $p = n+\lceil \frac{n}{2} \rceil$.
    The algorithm organizes every consecutive $p$ items as a group, and for each group, the $i$-th item is allocated to agent $j = \frac{2n+1}{2} - \left| i-\frac{2n+1}{2} \right|$.
    For example, for the first group (which contains items $\{ e_1,e_2,\ldots,e_p \}$), item $e_i \in \{e_1, e_2,\ldots, e_n\}$ is allocated to agent $i$; item $e_i\in \{ e_{n+1},\ldots,e_p\}$ is allocated to agent $2n+1-i$.

    Formally speaking, the algorithm allocates items solely depending on the index of the item: each item $e_j\in M$ will be allocated to agent $f(j)\in N$, where $f(j)  = \frac{2n+1}{2} - \left|(j-1 \mod p) - \frac{2n - 1}{2}\right|$.

	\begin{algorithm}[htb]
		\SetKwInOut{KwIn}{Input}\SetKwInOut{KwOut}{Output}
		\SetKw{Break}{Break}
		\caption{Sesqui-Round Robin Algorithm}\label{alg:monotone-chores} 
		Input: Monotone instance with $c_i(e_1) \geq c_i(e_2) \geq \cdots \geq c_i(e_m)$ for all $i\in N$\;
		Initialize: $X_i \gets \emptyset$ for all $i\in N$\;
		
		Let $p \gets n+\lceil \frac{n}{2} \rceil$ and define function $f(j)  = \frac{2n+1}{2} - \left|(j-1 \mod p) - \frac{2n - 1}{2}\right|$ \;
		\For{each online item $e_j\in M$}{
		    $i \gets f(j)$\;
		    $X_{i} \gets X_{i} + e_j$\;
		    }
		\KwOut{$\mathbf{X} = \{X_1,\cdots, X_n\}$}
	\end{algorithm}
    
	\begin{theorem}[Theorem 3.3 of \cite{aziz2022approximate}]
	    For any instance in which $c_i(e_1)\geq \cdots \geq c_i(e_m)$ for all agent $i\in N$, the algorithm Sesqui-Round Robin returns a $5/3$-MMS allocation.
	\end{theorem}
	
	Since the above algorithm does not require reallocation of any item, it applies in the online setting. Therefore, we immediately have the following. Moreover, since the allocation decision depends only on the function $f$, the algorithm runs in $O(m)$ time.
	
	\begin{corollary} \label{theorem:monotone-for-chores}
		For the allocation of chores, there exists a polynomial time $5/3$-competitive algorithm for monotone instances.
	\end{corollary}

    \subsection{Instances with Small Items}
    
    In this section, we consider small-item instances for the allocation of chores.
    
    \begin{definition}[Small-chore Instances]
    We call an online instance \emph{small-chore} instance if there exists some $\alpha < 1$ such that
    \begin{equation*}
        \forall i\in N, e\in M, c_i(e) \leq \alpha.
    \end{equation*}
    \end{definition}

    We show that there exists an algorithm that computes $(1+\alpha)$-MMS allocations for general number of agents.
    We further improve the result to $\sqrt{\alpha^2-4\alpha+5}+\alpha-1$ for the cases of two agents.
    Particularly, when $\alpha = 1$, the original ratio $\sqrt{2}$ is recovered; when $\alpha = 1/2$, the ratio is $\frac{\sqrt{13}-1}{2}\approx 1.30$.

    We first consider the case with $n$ agents and design an algorithm that allocates items greedily to active agents (see Algorithm~\ref{alg:small-chores}).
    Once an agent receives bundles of large cost we inactive this agent and when there is only one active agent we assign all future items to her.

	\begin{algorithm}[htb]
		\SetKwInOut{KwIn}{Input}\SetKwInOut{KwOut}{Output}
		\SetKw{Break}{Break}
		\caption{Algorithm-for-Small-Chore-Instances}\label{alg:small-chores} 
		Initialize: $A\gets N$ and for any $i\in N$, $X_i\gets \emptyset$\;
		\For{each online item $e\in M$}{
			\uIf{$|A|=1$}{
				$X_i\gets X_i + e$, where $i\in A$\;
				}
            \Else{
                $i\gets \argmin_{j\in A}\{c_j(e)\}$;
                $X_{i}\gets X_{i}+e$\;
				\If{$c_i(X_i) \geq 1$}{
					turn agent $i$ into inactive: $A\gets A\setminus \{i\}$\;
                    }
			}
            }
		\KwOut{$\mathbf{X} = (X_1, X_2,\cdots, X_n)$}
	\end{algorithm}

    \begin{theorem}\label{theorem:small-for-chores}
        For small-chore instances, Algorithm~\ref{alg:small-chores} computes a $(1+\alpha)$-MMS allocation in $O(m n)$ time.
    \end{theorem}
    
    \begin{proof}
    Note that throughout the execution of the algorithm, if $|A|\geq 2$, then all active agents $i\in A$ has $c_i(X_i) < 1$.
    For any agent $i$ that is inactivated, let $e_i$ be the last item agent $i$ receives, then we have
    \begin{equation*}
        c_i(X_i) < 1 + c_i(e_i) \leq 1 + \alpha \leq (1+\alpha) \cdot \MMS_i,
    \end{equation*}
    where the inequality follows from the fact that $\MMS_i\geq 1$ for the allocation of chores.
    Hence the final allocation is $(1+\alpha)$-MMS to all inactive agents.
    It remains to consider the case when $|A| = 1$ at the end of the algorithm.
    Let $i$ be the only active agent to which the last item is allocated.
    By the greedy allocation of the algorithm, for all $j\neq i$ we have $c_i(X_j) \geq c_j(X_j) \geq 1$, which implies
    \begin{align*}
        c_i(X_i) &= c_i(M) - \sum_{j\neq i} c_i(X_j) \leq n - (n-1) = 1 \leq \MMS_i.
    \end{align*}
    Hence, the allocation is also $(1+\alpha)$-MMS to agent $i$.
    Since the allocation of each item takes $O(n)$ time, the algorithm returns a $(1+\alpha)$-MMS allocation in $O(m n)$ time.
    \end{proof}

    In the following, we consider the case of two agents and improve the ratio to $\sqrt{\alpha^2-4\alpha+5} +\alpha -1$.
    The algorithm follows a similar idea from Algorithm~\ref{alg:general-2-agents} but with a different threshold that utilizes the additional information from the small-chore setting.
    
	\begin{algorithm}[htb]
		\SetKwInOut{KwIn}{Input}\SetKwInOut{KwOut}{Output}
		\SetKw{Break}{Break}
		\caption{Algorithm-for-Two-Agents-for-Small-Chore-Instances}\label{alg:small-2-agents} 
		Initialize: $X_1 \gets \emptyset$, $X_2 \gets \emptyset$\;
		\For{each online item $e\in M$}{
			$A \gets \{i\in N: c_i(X_i + e) \leq \gamma \}$ \;
			\uIf{$|A| = 1$}{
				$X_i \gets X_i + e$, where $i\in A$ \;
			}
			\uElseIf{$c_1(e) \leq \frac{2\gamma-2}{2-\gamma}\cdot c_2(e)$}{
				$X_1 \gets X_1 + e$\;
			}
			\Else{
				$X_2 \gets X_2 + e$\;
			}
		}
		\KwOut{$\mathbf{X} = (X_1, X_2)$}
	\end{algorithm}

    \begin{theorem}\label{theorem:small-2-agents}
        For small-chore instances with $n=2$ agents, Algorithm~\ref{alg:small-2-agents} computes a $\gamma$-MMS allocation in $O(m)$ time, where $\gamma = \sqrt{\alpha^2-4\alpha+5} +\alpha -1$.
    \end{theorem}
    
    \begin{proof}
    To show that the returned allocation is $\gamma$-MMS, it suffices to show that when each item $e$ arrives, the set $A$ is not empty.
    For the sake of contradiction, we assume that when some online item $e^*$ arrives we have $A = \emptyset$.
    That is, we have $c_i(X_i + e^*) > \gamma$ for both $i\in \{1,2\}$, where $X_i$ is the bundle agent $i$ holds right before $e^*$ arrives.
    We claim that under this situation, we have $c_1(e') > \frac{2\gamma-2}{2-\gamma} \cdot c_2(e')$ for all item $e'\in X_2$.
    \begin{claim}\label{claim:morethan}
        For all $e'\in X_2$, $c_1(e') > \frac{2\gamma-2}{2-\gamma} \cdot c_2(e')$.
    \end{claim}
    \begin{proof}
    Suppose otherwise and let $e_1$ be the first item assigned to $X_2$ with $c_1(e_1) \leq \frac{2\gamma-2}{2-\gamma} \cdot c_2(e_1)$.
    Let $X'_1$ be the bundle agent $1$ holds right before $e_1$ arrives.
    By the design of the algorithm, we must have $c_1(X'_1 + e_1) > \gamma$, because otherwise $e_1$ will be allocated to agent $1$.
    Recall that we assumed $c_1(X_1 + e^*) > \gamma$ for some item $e^*$ that is not allocated in the final allocation, which implies
    \begin{equation*}
        c_1(X_2) \leq c_1(M) - c_1(X_1 + e^*) < 2 - \gamma.
    \end{equation*}
    Therefore we have $c_1(e_1) \leq c_1(X_2) < 2-\gamma$ (because $e_1 \in X_2$), which gives
    \begin{equation*}
        c_1(X'_1) = c_1(X'_1 + e_1) - c_1(e_1) >  2\gamma - 2.
    \end{equation*}
    
    On the other hand, we show that $c_2(X'_1) > 2-\gamma$.
    Suppose otherwise, then there must exist some $e\in X'_1$ such that $c_1(e) > \frac{2\gamma - 2}{2-\gamma}\cdot c_2(e)$, which means that when item $e$ arrives we have $A = \{1\}$.
    Let $X''_1$ and $X''_2$ be bundles agent $1$ and $2$ hold right before $e$ arrives, respectively.
    We have $c_2(X''_2+e) > \gamma$ since $A = \{1\}$ when item $e$ arrives.
    Since $c_2(e) \leq \alpha$, we have $c_2(X''_2) \geq \gamma - \alpha$.
    Recall that $e_1$ is the first item assigned to $X_2$ with $c_1(e_1) \leq \frac{2\gamma-2}{2-\gamma} \cdot c_2(e_1)$, we have
    \begin{equation*}
        c_1(X''_2) > \frac{2\gamma-2}{2-\gamma} \cdot c_2(X''_2) \geq \frac{2\gamma-2}{2-\gamma} \cdot (\gamma-\alpha).
    \end{equation*}
    Note that $\gamma$ is the root of the equation $\frac{(2-x)^2}{2x-2}+\alpha = x$, we have
    \begin{equation*}
        c_1(X_1 + e^*) \leq 2 - c_1(X''_2) \leq 2 - \frac{2\gamma-2}{2-\gamma} \cdot (\gamma-\alpha) = \gamma,
    \end{equation*}
    which is a contradiction.
    Hence we have $c_2(X'_1) \geq c_1(X'_1)/(\frac{2\gamma-2}{2-\gamma}) > 2-\gamma$, which implies 
    $c_2(X_2 + e^*) \leq c_2(M - X'_1) < \gamma$, and it is a contradiction with our previous assumption that $c_2(X_2 + e^*) > \gamma$ for some item $e^*$ (that arrives after $e_1$).
    \end{proof}
    Recall that $c_1(X_2) \leq c_1(M) - c_1(X_1 + e^*) < 2- \gamma$.
    By Claim~\ref{claim:morethan}, we have
    \begin{equation*}
        c_2(X_2) < c_1(X_2)/(\frac{2\gamma-2}{2-\gamma}) < \frac{(2-\gamma)^2}{2\gamma-2}.
    \end{equation*}
    Note that $\gamma$ is the root of the equation $\frac{(2-x)^2}{2x-2}+\alpha = x$, we have $c_2(X_2 + e^*) \leq \frac{(2-\gamma)^2}{2\gamma-2} + \alpha = \gamma$,
    which is a contradiction.
    Finally, since the allocation of each item takes $O(1)$ time, the algorithm returns a $\gamma$-MMS allocation in $O(m)$ time.
    \end{proof}

    \section{Experimental Evaluation}
	
	In this section, we perform the experimental evaluation of our algorithms on several generated datasets.
	For a more meaningful demonstration, we only consider the case with $n\geq 3$ agents.
    Therefore, we consider our $0.5$-competitive algorithm for monotone instances for the allocation of goods and the $(2-1/n)$-competitive algorithm for the allocation of chores.
	We compare the approximation ratios of the allocations returned by our algorithms and the greedy algorithm, which greedily allocates each item $e$ to the agent with maximum value (resp. minimum cost) on $e$.
	We generate random datasets with different sizes.
	The small datasets contain $10$ agents and $100$ items; the large datasets contain $50$ agents and $500$ items.
	For each size, we randomly generate 1000 instances.
	Recall that for any $i\in N$, we normalize $v_i(M) = n$ for goods and $c_i(M) = n$ for chores.
	To ensure that $\MMS_i = 1$ for any $i\in N$, we adopt the following generating method.
	Imagine that each agent has one item of value/cost $n$.
	We break this item into $m$ items by first dividing it into $n$ items, each with value/cost $1$, and then randomly break each of these $n$ items into $m/n$ items.
%
	We sort the generated $m$ items in descending order of values/costs.
	Hence, for all agent $i\in N$, we have $v_i(e_1) \geq v_i(e_2) \geq \cdots \geq v_i(e_m)$ and $\MMS_i = 1 $ for goods, $c_i(e_1) \geq c_i(e_2) \geq \cdots \geq c_i(e_m)$ and $\MMS_i = 1$ for chores.
    In experiments for goods, the items arrive in the order of $\{e_1,\ldots,e_m\}$, i.e., as in the monotone instances.
	For the allocation of chores, the arrival order is random.
	
	For the allocation of goods, we compare the performances of the greedy algorithm and our algorithm in Table~\ref{tab:mono_goods}.
	We report the average and minimum approximation ratios returned by the greedy algorithm and our algorithm over the 1000 instances as well as the total running time for both algorithms.
	It can be seen that the performance of the greedy algorithm is significantly worse than our algorithm in both small and large datasets.
	The average approximation ratios returned by the greedy algorithm in small and large datasets are $0.113$ and $0$, respectively, while those returned by our algorithm are $0.532$ and $0.514$, respectively.
	Indeed, the greedy algorithm returns a $0$-MMS allocation for more than 40\% instances, for small dataset. In large dataset, the greedy algorithm returns a $0$-MMS allocation for more than 99\% instances.
	In contrast, the approximation ratios of the allocations returned by our algorithm under both datasets are at least $0.5$, which matches our theoretical guarantee for monotone instances (\cref{theorem:monotone-for-goods}). 
	Furthermore, our algorithm has a slight advantage in running time due to inactivation operation, which reduces the inactive agents for receiving items.
%
%
	This set of experiments demonstrate the effectiveness of our algorithm, which computes an approximately MMS allocation with an approximation ratio at least that from the theoretical guarantee ($0.5$) in each dataset.
	
	\begin{table}[htbp]
	    \centering
	    \begin{tabular}{c||c|c|c|c|c|c}
	    \hline \hline
	         \multirow{2}{*}{Algorithm} &  \multicolumn{3}{c|}{Small} & \multicolumn{3}{c}{Large} \\ \cline{2-7}
	         & Avg. & Min. & Time (s) & Avg. & Min. & Time (s)\\
	         \hline \hline
	         Greedy & 0.113 & 0 & 0.13 & 0 & 0 & 2.05 \\ \hline
	         Ours & 0.532 & 0.5 & 0.09 & 0.514 & 0.5 & 0.93 \\ \hline \hline
	    \end{tabular}
	    \caption{Approximation ratios of the allocations returned by different algorithms for goods in different cases. (Min. means the minimum ratio and Avg. means the average ratio, over the 1000 instances). We also report the total running time (in seconds) for computing the allocations by each algorithm.}
	    \label{tab:mono_goods}
	\end{table}

    For the allocation of chores, the results are shown in Table~\ref{tab:chores}.
    We report the average and maximum approximation ratios returned by the greedy algorithm and our algorithm over the 1000 instances as well as the total running time for both algorithms.
    Recall from our theoretical analysis that the competitive ratio of our algorithm is at most $1.9$ for small dataset and $1.99$ for large dataset.
    From Table~\ref{tab:chores}, we observe that the approximation ratios of the allocations returned by our algorithm under both datasets are around $1.2$, and the maximum approximation ratios ($1.28$ for the small dataset and $1.39$ for the large dataset) are evidently better than the theoretical guarantees.
    In comparison, the average approximation ratio of the greedy algorithm is $2.33$ in the small dataset and grows to $7.96$ in the large dataset.
    The difference between our algorithm and the greedy algorithm is that our algorithm only allocates items to the active agents, i.e., those who receive items with cost less than the threshold ($c_i(X_i)\leq 1-1/n$).
    The experiment shows that this is crucial in achieving good approximation ratios in practice.
    
  	\begin{table}[htbp]
	    \centering
	    \begin{tabular}{c||c|c|c|c|c|c}
	    \hline \hline
	         \multirow{2}{*}{Algorithm} &  \multicolumn{3}{c|}{Small} & \multicolumn{3}{c}{Large} \\ \cline{2-7}
	         & Avg. & Max. & Time (s) & Avg. & Max. & Time (s) \\
	         \hline \hline
	         Greedy & 2.33 & 4.42 & 0.13 & 7.96 & 17.99 & 0.26 \\ \hline
	         Ours & 1.09 & 1.28 & 2.2 & 1.26 & 1.39 & 4.8 \\ \hline \hline
	    \end{tabular}
	    \caption{Approximation ratios of the allocations returned by different algorithms for chores in different cases. (Max. means the maximum ratio and Avg. means the average ratio in 1000 examples)}
	    \label{tab:chores}
	\end{table}
	
    To present the performance difference between the greedy algorithm and our algorithm in more detail, we choose the small dataset for chores (for which the greedy algorithm performs the best).
    We plot the cumulative distribution function of approximation ratios returned by the two algorithms in Figure~\ref{fig:chores_cdf}.
    It can be seen that our algorithm guarantees an approximation ratio smaller than $2$ for all instances, while only around $20\%$ instances returned by the greedy algorithm have an approximation ratio no more than $2$.
    For about half of the instances, the greedy algorithm computes an allocation with an approximation ratio exceeding $2.5$ with respect to MMS.

  	\begin{figure}[htbp]
	    \centering
        \includegraphics[width=0.45\textwidth]{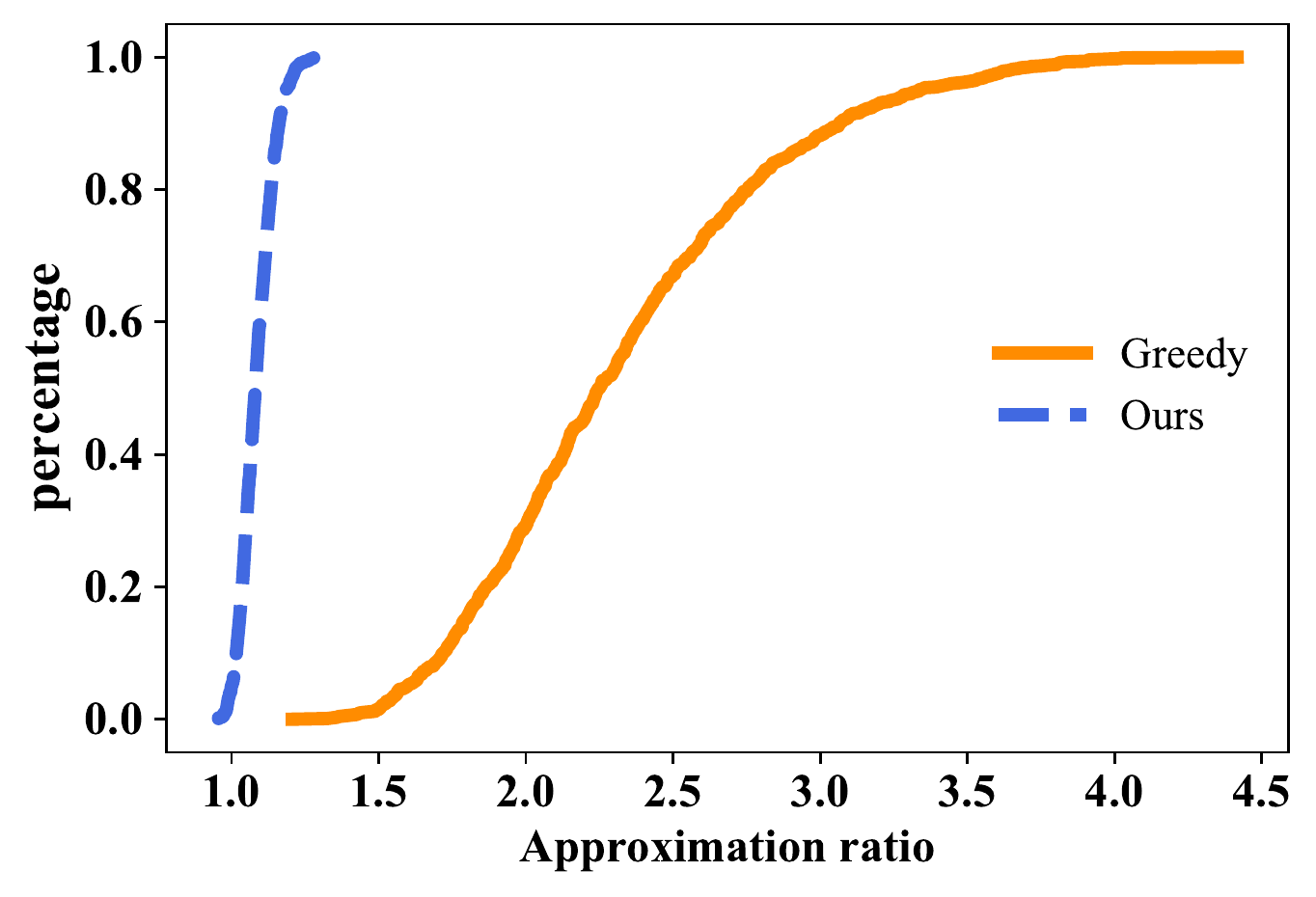}
        \caption{Cumulative distribution function of the approximation ratios of allocations returned by the two algorithms.}\label{fig:chores_cdf}
	\end{figure}

	\section{Conclusion and Open Problems}
	
	In this paper, we study the problem of fairly allocating indivisible online items to a group of agents with general additive valuation functions.
	For the allocation of goods, we show that no algorithm can guarantee any non-zero competitive ratio for $n\geq 3$ agents and propose an optimal $0.5$-competitive algorithm for two agents.
	For the allocation of chores, we propose a $(2-1/n)$-competitive algorithm for $n\geq 3$ agents, a $\sqrt{2}$-competitive algorithm for two agents, and show that no algorithm can do better than $15/11$-competitive for two agents.
	We also study the monotone instances for both goods and chores and improve the competitive ratios to $0.5$ and $5/3$ for goods and chores, respectively.
    We further consider the small items instances such that all value/cost are no larger than $\alpha$ where $(1-\alpha)$-competitive algorithm and $(1+\alpha)$-competitive algorithm exist for goods and chores respectively.
    Furthermore, we improve the competitive ratio to $\sqrt{\alpha^2-4\alpha+5}+\alpha-1$ when allocating small chores to two agents.
	
	There are many open problems regarding the online approximation of MMS allocations.
	First, while we show that no competitive algorithm exists for the online allocation of goods and competitive algorithms exist for monotone instances, we are curious about a less restrictive condition under which competitive algorithm exists.
	Second, for the allocation of chores, the lower bound $1.585$ for general number of agents follows from the identical valuation case, and it remains unknown whether the general additive valuation case is strictly harder.
	It is also interesting to investigate the optimal competitive ratio (which is in $[1.585,2)$) for general number of agents and that for the case of two agents (which is in $[1.364,1.414]$).
	Finally, while we show that normalization of the valuation functions is necessary to achieve a bounded ratio for the allocation of goods, it remains unknown whether constant competitive ratios can be achieved for chores if the cost functions are not normalized.

\bibliographystyle{abbrv}
\bibliography{online-mms}

\newpage
\appendix

    \section{Justifications on the Normalization Assumption}
    \label{sec:justification}
    
    In this section, we justify the normalization assumption of our model.
    We show that without the normalization assumption, even for the case of $n=2$,
    \begin{itemize}
        \item no algorithm can do strictly better than $0$-competitive for the allocation of goods;
        \item no algorithm can do strictly better than $2$-competitive for the allocation of chores.
    \end{itemize}
    
    We first show the hard instance for the allocation of goods.
    Recall from \cref{theorem:goods-n-agents} that with the normalization assumption, no algorithm can do strictly better than $0$-competitive when $n\geq 3$.
    Thus it suffices to consider the case when $n=2$.
    
    \begin{theorem}
        For the allocation of goods, without the normalization assumption, no online algorithm has competitive ratio strictly larger than $0$ when $n=2$.
    \end{theorem}
    
    \begin{proof}
    Assume otherwise and there exists a $\gamma$-competitive algorithm for approximating MMS allocation for $n=2$ agents, where $\gamma \in (0,1]$.
    Let $r$ be a sufficiently large integer such that $\gamma > 1/r$.
    We show that in the allocation returned by the algorithm for the following instance, at least one agent $i\in N$ is allocated a bundle $X_i$ with $v_i(X_i) < \gamma\cdot \MMS_i$, which is a contradiction.
    
    Let the first item be $e_1$ with $v_1(e_1) = v_2(e_1) = 1$ and assume w.l.o.g. that agent $1$ receives it.
    Then let the second item $e_2$ with $v_1(e_2) = r, v_2(e_2) = 1/r$.
    Since the algorithm is $\gamma$-competitive, $e_2$ should be assigned to agent $2$ as otherwise for the instance with only two items $\{e_1,e_2\}$, we have $X_2 = \emptyset$, which leads to a $0$-MMS allocation.
    Let $X_i'$ be the bundle agent $i$ holds before item $e_3$ comes.
    
    \begin{table}[htb]
		\centering
			\begin{tabular}{ c|c|c|c } 
				$\qquad$ & $e_1$ & $e_2$ & $e_3$  \\
				\hline
				$\mathbf{1}$ & \boxed{$1$} & $r$ & $r$ \\ 
				\hline
				$\mathbf{2}$ & $1$ & \boxed{$1/r$} & $1$
			\end{tabular}
		\caption{Hard instance for goods without the normalization assumption.}\label{tab:Hardness-justify}
	\end{table}

    For the instance shown in Table~\ref{tab:Hardness-justify}, the last item $e_3$ holds $v_1(e_3) = r, v_2(e_3) = 1$.
    We have 
    \begin{equation*}
        \MMS_1 = r, \MMS_2 = 1,\quad \text{and} \quad
        v_1(X_1') = 1, v_2(X_2') = 1/r,
    \end{equation*}
    where $v_1(X_1') = (1/r)\cdot \MMS_1 < \gamma\cdot \MMS_1, v_2(X_2') = (1/r)\cdot \MMS_2 < \gamma\cdot \MMS_2$.
    Since there must exist one agent in $\{1,2\}$ that does not receive item $e_3$, the final allocation is not $\gamma$-MMS, which is a contradiction.
    \end{proof}
    
    Next, we show that without the normalization assumption, no online algorithm can guarantee a competitive ratio strictly smaller than $2$ when $n=2$.
	Note that any allocation is trivially $2$-MMS when $n=2$.
	
    \begin{theorem}
        For the allocation of chores, without the normalization assumption, no online algorithm has competitive ratio strictly smaller than $2$ when $n=2$.
    \end{theorem}
    \begin{proof}
	Assume otherwise and there exists an online algorithm with a competitive ratio $2-\gamma$ for some constant $\gamma > 0$.
	Let $\epsilon$ be sufficiently small such that $\gamma > 2\epsilon$.
    We first construct the following instance (see Table~\ref{tab:Hardness-justify-chores1}), and show that the algorithm can not allocate all items $e_1, \cdots, e_l$ to agent $1$, for $l = \lceil 1/\epsilon \rceil$.

    \begin{table}[htb]
		\centering
			\begin{tabular}{ c|c|c|c|c|c|c } 
				$\qquad$ & $e_1$ & $e_2$ & $e_3$ & $e_4$ & $\cdots$ & $e_l$\\
				\hline
				$\mathbf{1}$ & $\boxed{1}$ & $\epsilon$ & $\epsilon$ & $\epsilon$ & $\cdots$ & $\epsilon$\\ 
				\hline
				$\mathbf{2}$ & $1$ &  $\epsilon^{-1}$ & $\epsilon^{-2}$ & $\epsilon^{-3}$ &
				$\cdots$ & $\epsilon^{-(l-1)}$

			\end{tabular}
		\caption{Hard instance for chores without the normalization assumption ($n=2$).}\label{tab:Hardness-justify-chores1}
	\end{table}
	
    Specifically, the first item $e_1$ has $c_1(e_1) = c_2(e_1) = 1$, and we can assume w.l.o.g. that agent $1$ receives it.
    For all $2\leq i \leq l$, we have
    \begin{equation*}
        c_1(e_i) = \epsilon, \quad c_2(e_i) = \epsilon^{-(i-1)}.
    \end{equation*}
    For the sake of contradiction, we assume all items $\{e_1, \cdots, e_l\}$ are allocated to agent $1$.
    Then for the instance with exactly $l$ items, we have
    \begin{equation*}
        c_1(X_1) = 1 + (\lceil 1 / \epsilon \rceil-1) \cdot \epsilon \geq 2-\epsilon, \quad
        \MMS_1 = 1,
    \end{equation*}
    which contradicts the assumption that the algorithm has a competitive ratio $2-\gamma$.
    Hence the algorithm must allocate some items in $\{e_2, \cdots, e_l\}$ to agent $2$.
    Let $e_p$ be the first item allocated to agent $2$, where $1 < p \leq l$.
    Then we consider another instance whose first $p$ items are the same as in Table~\ref{tab:Hardness-justify-chores1}, but has different future items.
    Let the next (and last) item be $e_{p+1}$ (see Table~\ref{tab:Hardness-justify-chores2}).
    
    \begin{table}[htbp]
		\centering
			\begin{tabular}{ c|c|c|c|c|c|c}
				$\qquad$ & $e_1$ & $e_2$ & $\cdots$ & $e_{p-1}$ & $e_p$ & $e_{p+1}$ \\
				\hline
				$\mathbf{1}$ & $\boxed{1}$ & $\boxed{\epsilon}$ & $\cdots$ & $\boxed{\epsilon}$ & ${\epsilon}$ & $1+(p-1)\cdot\epsilon$ \\ 
				\hline
				$\mathbf{2}$ & $1$ &  $\epsilon^{-1}$ & $\cdots$ & $\epsilon^{-(p-2)}$ & $\boxed{\epsilon^{-(p-1)}}$ & $\epsilon^{-(p-1)}$ 
			\end{tabular}
		\caption{Hard instance for chores without the normalization assumption ($n=2$).}\label{tab:Hardness-justify-chores2}
    \end{table}
    
    Let $X_i$ be the bundle that agent $i$ holds before item $e_{p+1}$ comes.
    Then we have
    \begin{align*}
        c_1(X_1) &= 1+(p-2)\cdot\epsilon, \quad c_2(X_2)  = \epsilon^{-(p-1)}, \\
        \MMS_1 &= 1+(p-1)\cdot\epsilon, \quad \MMS_2 = \epsilon^{-(p-1)}+\epsilon^{-(p-2)}.
    \end{align*}
    Hence we have
    \begin{align*}
        c_1(X_1+e_{p+1}) & > (2-\epsilon)\cdot\MMS_1, \\ c_2(X_2+e_{p+1}) &= 2\cdot\epsilon^{-(p-1)} = \frac{2}{1+\epsilon}\cdot \MMS_2 > (2-2\epsilon)\cdot\MMS_2.
    \end{align*}
    The last equality holds because $1 > 1-\epsilon^2$.
    Hence no matter which agent in $\{1,2\}$ receives item $e_{p+1}$, the allocation is not $(2-2\epsilon)$-MMS to her.
	\end{proof}

    The main idea of our construction are: 1) The first item $e_1$ assigned to the agent $1$, costs exactly $\MMS_1$ to her at this moment, and each following coming item costs $\epsilon\cdot\MMS_1$ to her;
    2) For agent $2$, the cost of coming item increases exponentially, i.e. $c_2(e_i) = \epsilon^{-(i-1)}$ until she receives her first item.
    For agent $1$, although the items (other than the first one) cost extremely small to her, she cannot receive all of them because otherwise, the allocation becomes $2$-MMS. 
    Hence the algorithm has to allocate at least one item to agent $2$ before the $\lceil 1/\epsilon \rceil$-th item arrives.
    Then consider the first item $e_p$ that agent $2$ receives.
    Then we construct another instance in which the next item $e_{p+1}$ satisfies $c_i(X_i+e_{p+1}) \geq (2-2\epsilon)\cdot\MMS_i$ for both $i\in N$, for which the algorithm can not guarantee a competitive ratio of $2-\gamma$.
    
	Following this idea, we can extend our result to $n\geq 3$ agents straightforwardly.
	Since the proof is almost identical (except that we need to construct more items), we only provide a proof sketch here.
	
    \begin{theorem}
        For the allocation of chores with $n\geq 3$ agents, without the normalization assumption, no online algorithm has competitive ratio strictly smaller than $2$.
    \end{theorem}
	\begin{proof}[Proof Sketch]
	We construct a collection of instances that maintain the following properties: 1) The first item that each agent $i$ receives costs exactly $\MMS_i$ to her, and each of the following items costs $\epsilon\cdot\MMS_i$ to her; 2) For those agents who have not received any item, the cost of coming items increase exponentially.
	We call that the first item each agent receives a large item to her, and after that each item is small.
%
	Note that any agent can receive at most one large item and $\lceil 1/\epsilon \rceil-2$ small items,
	as otherwise the items she received cost more than $(2-2\epsilon)\cdot\MMS_i$ to her.
	Hence $n-1$ agents receive at most $(n-1)\cdot (\lceil 1/\epsilon \rceil - 1)$ items, and the algorithm should allocate each agent at least one item before the $\lceil (n-1)/\epsilon \rceil$-th item arrives.
%
	Consider the time when each agent has received at least one item, i.e., $X_i\neq \emptyset$ for all $i\in N$.
	Then we can construct the last item (similar to $e_{p+1}$ in Table~\ref{tab:Hardness-justify-chores2}) such that no matter which agent receives it, the allocation is not $(2-2\epsilon)$-MMS to her, e.g., by making $c_i(e_{p+1})\approx c_i(X_i)$ for all $i\in N$.
	\end{proof}

    \section{Missing Proofs of Theorem \ref{theorem:goods-n-agents}} \label{sec:hardness-goods->=4}

    In this section we complete the proof of \cref{theorem:goods-n-agents} when there are $n\geq 4$ agents.
    Similar to our result for $n=3$, we construct a collection of instances showing that no online algorithm can guarantee a competitive ratio strictly larger than 0 on these instances.
    
\begin{lemma}{\label{lemma:goods-n-agents-appendix}}
    No online algorithm has competitive ratio strictly larger than 0 for approximating MMS allocations for goods, for any $n\geq 4$.
\end{lemma}
\begin{proof}
    For the sake of contradiction, suppose there exists a $\gamma$-competitive algorithm for approximating MMS allocation for $n\geq 4$ agents, where $\gamma\in (0,1]$.
    Let $r > 1/\gamma$ be a sufficiently large integer and $\epsilon>0$ be sufficiently small such that $r^3\epsilon < \gamma$.
    We construct a collection of instances and show that for the allocation returned by the algorithm for at least one of these instances, at least one agent $i\in N$ is allocated a bundle $X_i$ with $v_i(X_i) < (1/r)\cdot \MMS_i$, which is a contradiction.
	Recall that for deterministic allocation algorithms, we can construct an instance adaptively depending on how the previous items are allocated.
	
	Before we construct the instances, we first show an observation that if at some time $t$ there are $k$ agents with $v_i(X_i') < (1/r)\cdot \MMS_i$ (where $X_i'$ is the bundle agent $i$ holds at time $t$) while there are less than $k$ items to be arrived, then the allocation returned by the algorithm will not be $\gamma$-MMS.
	
	\begin{observation} \label{observation:items-less-than-agents}
	    If at some time $t$ during the allocation, there are $k$ agents $K$ such that for all $i\in K$ we have $v_i(X_i') < (1/r)\cdot \MMS_i$, and there are at most $k-1$ items to be arrived, then the final allocation returned by the algorithm is not $\gamma$-MMS.
	\end{observation}
	\begin{proof}
	    Since there are at most $k-1$ items that arrive after time $t$, some agent $i\in K$ will receive no item after time $t$, i.e. $X_i = X_i'$.
	    Hence in the final allocation we have $v_i(X_i) < 1/r\cdot \MMS_i$ for some $i\in K$ and the allocation returned by the algorithm is not $\gamma$-MMS.
	\end{proof}
	
	Based on this observation we argue that any algorithm with a competitive ratio greater than $0$ should assign the first $n-1$ items to $n-1$ different agents (unless some agent has value $0$ on some item).
	
	\begin{claim}
	    Assume that for all agent $i\in N$ and item $e\in M$ we have $v_i(e) > 0$, and $m\geq n$.
	    Then any algorithm with competitive ratio greater than $0$ should assign the first $n-1$ items $\{e_1, \cdots, e_{n-1}\}$ to $n-1$ different agents.
	\end{claim}
    \begin{proof}
        Assume otherwise and some items $e_i, e_j$ ($i\neq j$) are assigned to the same agent.
        We can assume w.l.o.g. that agent $n-1$ and $n$ receive nothing at the time when item $e_n$ arrives, i.e. $X_{n-1}' = X_n' = \emptyset$.
        We show that the algorithm fails in the instances with $n$ items.
        Let the last item be $e_n$ such that $v_i(e_n) = v_i(M) - \sum_{j=1}^{n-1} v_i(e_j) > 0$ for all $i\in N$.
        Then for both agents $n-1$ and $n$ we have $v_{n-1}(X_{n-1}') = v_n(X_n') = 0$, while $\MMS_{n-1} > 0, \MMS_n > 0$.
        Since there is only one item $e_n$ to be allocated, from Observation~\ref{observation:items-less-than-agents}, the allocation returned by the algorithm is not $\gamma$-MMS, which is a contradiction.
    \end{proof}
    
    Next, we construct the instances for which no online algorithm can be $\gamma$-competitive.
	For each item $e_j$, $1 \leq j \leq n-1$, we let (see Table~\ref{tab:General-Case0}) $v_i(e_{j}) = r^3\epsilon$ for all $i < j$ and $v_i(e_{j}) = \epsilon$ for all $i \geq j$.
	
	Recall that the first $n-1$ items must be assigned to $n-1$ different agents.
	Since when item $i$ arrives, the valuation functions of agents $\{i, i+1,\ldots, n\}$ are identical, we assume w.l.o.g. that agent $i$ receives item $e_i$ for all $i \in \{1, 2, \cdots, n-1\}$.
	Then we consider the next item $e_n$ with
	\begin{equation*}
        v_i(e_{n}) = 
        \begin{cases}
	        r\epsilon, \quad\ \text{ if } i=1 \\
	        r^3\epsilon, \quad \text{if } 2 \leq i \leq n-1 \\
	        \epsilon, \qquad \text{if } i=n.
        \end{cases}
	\end{equation*}
	
	We first show that any $\gamma$-competitive algorithm must assign $e_n$ to an agent in $\{1,n\}$.
	Assume otherwise, i.e., $e_n\in X_i$ for some $i\notin\{1,n\}$.
	Then we consider the instance with $n+1$ items (as shown in Table~\ref{tab:General-Case0}), where the last item $e_{n+1}$ has
	\begin{equation*}
         v_i(e_{n+1}) = 
        \begin{cases}
	       3-(n-2)r^3\epsilon-r\epsilon-\epsilon, \quad \text{ if } i=1 \\
	       3-(n-i)r^3\epsilon-i\epsilon, \qquad\  \quad \text{if } 2 \leq i \leq n \\
	       3-(n+1)\epsilon, \qquad \qquad \qquad \text{if } i = n.\\
        \end{cases}
	\end{equation*}

	\begin{table}[htb]
       \centering
		\begin{tabular}{ c|c|c|c|c|c|c|c } 
			$\qquad$ & $e_1$ & $e_2$ & $e_3$ & $\cdots$ & $e_{n-1}$ & $e_n$ & $e_{n+1}$\\
			\hline
			$\mathbf{1}$ & $\boxed{\epsilon}$ & $r^3\epsilon$ & $r^3\epsilon$ & $\cdots$ & $r^3\epsilon$ & $r\epsilon$ &$3 - (n-2)r^3\epsilon - r\epsilon-\epsilon$\\ 
			\hline
			$\mathbf{2}$ & $\epsilon$ & $\boxed{\epsilon}$ & $r^3\epsilon$ & $\cdots$ & $r^3\epsilon$ & $r^3\epsilon$ & $3-(n-2)r^3\epsilon-2\epsilon$\\
			\hline
			$\mathbf{3}$ & $\epsilon$ & $\epsilon$ & $\boxed{\epsilon}$ & $\cdots$ & $r^3\epsilon$ & $r^3\epsilon$ &$3-(n-3)r^3\epsilon-3\epsilon$\\
			\hline
			$\cdots$ & $\cdots$ & $\cdots$ & $\cdots$ & $\cdots$ & $\cdots$  & $\cdots$ & $\cdots$\\
			\hline
			$\mathbf{n-1}$ & $\epsilon$ & $\epsilon$ & $\epsilon$ & $\cdots$ & $\boxed{\epsilon}$ & $r^3\epsilon$ & $3-r^3\epsilon-(n-1)\epsilon$\\
			\hline
			$\mathbf{n}$ & $\epsilon$ & $\epsilon$ & $\epsilon$ & $\cdots$ & $\epsilon$ & $\epsilon$ & $3-n\epsilon$\\
		\end{tabular}
    	\caption{Instance showing why item $e_n$ must be assigned to one of the agents in $\{1,2\}$.}\label{tab:General-Case0}
	\end{table}
	
	Observe that we have $X_1' = \{e_1\}, X_n' = \emptyset$ and $\MMS_1 = \epsilon+r\epsilon, \MMS_n > 0$ while there is only one item $e_{n+1}$ that is not allocated.
	From Observation~\ref{observation:items-less-than-agents}, we have a contradiction.
	Hence every $\gamma$-competitive algorithm must allocate item $e_{n}$ to either agent $1$ or $n$.
	Depending on which agent receives item $e_{n}$, we construct two different instances.
	We argue that for both instances, the allocation returned by the algorithm is not $\gamma$-MMS.
	
	For the case when agent $1$ receives item $e_n$, we construct the instance with $n+2$ items as in Table~\ref{tab:General-Case1}.
	
	\begin{table}[htb]
       \centering
		\begin{tabular}{ c|c|c|c|c|c|c|c|c } 
			$\qquad$ & $e_1$ & $e_2$ & $e_3$ & $\cdots$ & $e_{n-1}$ & $e_n$ & $e_{n+1}$ & $e_{n+2}$\\
			\hline
			$\mathbf{1}$ & $\boxed{\epsilon}$ & $r^3\epsilon$ & $r^3\epsilon$ & $\cdots$ & $r^3\epsilon$ & $\boxed{r\epsilon}$ & $r^3\epsilon$ &$3 - (n-1)r^3\epsilon - r\epsilon-\epsilon$\\ 
			\hline
			$\mathbf{2}$ & $\epsilon$ & $\boxed{\epsilon}$ & $r^3\epsilon$ & $\cdots$ & $r^3\epsilon$ & $r^3\epsilon$ & $r^3\epsilon$ & $3-(n-1)r^3\epsilon-2\epsilon$\\
			\hline
			$\mathbf{3}$ & $\epsilon$ & $\epsilon$ & $\boxed{\epsilon}$ & $\cdots$ & $r^3\epsilon$ & $r^3\epsilon$ & $r^3\epsilon$ &$3-(n-2)r^3\epsilon-3\epsilon$\\
			\hline
			$\cdots$ & $\cdots$ & $\cdots$ & $\cdots$ & $\cdots$ & $\cdots$  & $\cdots$ & $\cdots$ & $\cdots$\\
			\hline
			$\mathbf{n-1}$ & $\epsilon$ & $\epsilon$ & $\epsilon$ & $\cdots$ & $\boxed{\epsilon}$ & $r^3\epsilon$ & $r^3\epsilon$ & $3-2r^3\epsilon-(n-1)\epsilon$\\
			\hline
			$\mathbf{n}$ & $\epsilon$ & $\epsilon$ & $\epsilon$ & $\cdots$ & $\epsilon$ & $\epsilon$ &$\epsilon$ & $3-(n+1)\epsilon$\\
		\end{tabular}
    	\caption{Instance with $n+2$ items when item $e_n$ is assigned to agent $1$.}\label{tab:General-Case1}
	\end{table}
	Note that for any $i\in \{1,2,n\}$, we have $v_i(X_i') < (1/r)\cdot \MMS_i$, where $X_i'$ is the bundle agent $i$ holds before $e_{n+1}$ arrives, because
	\begin{align*}
	    \MMS_1 = \MMS_2 = r^3\epsilon, \quad \MMS_n = \epsilon, \\
	    v_1(X_1') = r\epsilon+\epsilon, \quad v_2(X_2') = \epsilon, \quad v_n(X_n') = 0.
	\end{align*}
	Since there are only two items $\{e_{n+1}, e_{n+2}\}$ to be allocated, from Observation~\ref{observation:items-less-than-agents}, the allocation returned by the algorithm is not $\gamma$-MMS.
	
	Next we consider the case when agent $n$ receives item $e_n$ and let the next item be $e_{n+1}$ with
	\begin{equation*}
	    v_i(e_{n+1}) = 
        \begin{cases}
	        \epsilon, \qquad \text{if } i \in \{1,n\} \\
	        r\epsilon, \quad\ \text{ if } i=2 \\
	        r^3\epsilon, \quad \text{if } 3 \leq i \leq n-1.
        \end{cases}
	\end{equation*}
	We argue that item $e_{n+1}$ must be assigned to agent $1$ or $2$.
	Assume otherwise, i.e., item $e_{n+1}$ is assigned to some agent $j\notin\{1,2\}$.
	For any $i\in N$, let $X_i'$ be the bundle agent $i$ holds before item $e_{n+2}$ comes.
	For the instance with $n+2$ items shown in Table~\ref{tab:General-Case2a}, we have 
	\begin{equation*}
	    \MMS_1 = \MMS_2 = r\epsilon +2\epsilon, \quad
	    v_1(X_1') = v_2(X_2') = \epsilon,
	\end{equation*}
	where $v_1(X_1') < (1/r)\cdot \MMS_1$ and $v_2(X_2') < (1/r)\cdot \MMS_2$.
	Since there must exist at least one agent in $\{1,2\}$ that does not receive item $e_{n+2}$, we have $X_1 = X_1'$ or $X_2 = X_2'$ in the final allocation.
	Therefore the allocation is not $\gamma$-MMS.
	
	\begin{table}[htb]
       \centering
		\begin{tabular}{ c|c|c|c|c|c|c|c|c } 
			$\qquad$ & $e_1$ & $e_2$ & $e_3$ & $\cdots$ & $e_{n-1}$ & $e_n$ & $e_{n+1}$ & $e_{n+2}$\\
			\hline
			$\mathbf{1}$ & $\boxed{\epsilon}$ & $r^3\epsilon$ & $r^3\epsilon$ & $\cdots$ & $r^3\epsilon$ & $r\epsilon$ & $\epsilon$ & $3 - (n-2)r^3\epsilon - r\epsilon-2\epsilon$\\ 
			\hline
			$\mathbf{2}$ & $\epsilon$ & $\boxed{\epsilon}$ & $r^3\epsilon$ & $\cdots$ & $r^3\epsilon$ & $r^3\epsilon$ & $r\epsilon$ & $3-(n-2)r^3\epsilon-r\epsilon-2\epsilon$\\
			\hline
			$\mathbf{3}$ & $\epsilon$ & $\epsilon$ & $\boxed{\epsilon}$ & $\cdots$ & $r^3\epsilon$ & $r^3\epsilon$ & $r^3\epsilon$ &$3-(n-2)r^3\epsilon-3\epsilon$\\
			\hline
			$\cdots$ & $\cdots$ & $\cdots$ & $\cdots$ & $\cdots$ & $\cdots$  & $\cdots$ & $\cdots$ & $\cdots$\\
			\hline
			$\mathbf{n-1}$ & $\epsilon$ & $\epsilon$ & $\epsilon$ & $\cdots$ & $\boxed{\epsilon}$ & $r^3\epsilon$ & $r^3\epsilon$ & $3-2r^3\epsilon-(n-1)\epsilon$\\
			\hline
			$\mathbf{n}$ & $\epsilon$ & $\epsilon$ & $\epsilon$ & $\cdots$ & $\epsilon$ & $\boxed{\epsilon}$ &$\epsilon$ & $3-(n+1)\epsilon$\\
		\end{tabular}
    	\caption{Instance showing why item $e_{n+1}$ must be assigned to one of the agents in $\{1,2\}$.}\label{tab:General-Case2a}
    \end{table}
    
	Hence the algorithm must allocate item $e_{n+1}$ to agent $1$ or $2$, for which case we construct the instance with $n+3$ items shown in Table~\ref{tab:General-Case2b}.
	Note that for this instance we have
	\begin{equation*}
	    \MMS_1 = \MMS_2 = \MMS_3 = r^3\epsilon.
	\end{equation*}
	
	\begin{table}[htb]
       \centering
		\begin{tabular}{ c|c|c|c|c|c|c|c|c|c } 
			$\qquad$ & $e_1$ & $e_2$ & $e_3$ & $\cdots$ & $e_{n-1}$ & $e_n$ & $e_{n+1}$ & $e_{n+2}$ & $e_{n+3}$\\
			\hline
			$\mathbf{1}$ & $\boxed{\epsilon}$ & $r^3\epsilon$ & $r^3\epsilon$ & $\cdots$ & $r^3\epsilon$ & $r\epsilon$ & $\epsilon$ & $r^3\epsilon$ & $3 - (n-1)r^3\epsilon - r\epsilon-2\epsilon$\\ 
			\hline
			$\mathbf{2}$ & $\epsilon$ & $\boxed{\epsilon}$ & $r^3\epsilon$ & $\cdots$ & $r^3\epsilon$ & $r^3\epsilon$ & $r\epsilon$ & $r^3\epsilon$ & $3-(n-1)r^3\epsilon-r\epsilon-2\epsilon$\\
			\hline
			$\mathbf{3}$ & $\epsilon$ & $\epsilon$ & $\boxed{\epsilon}$ & $\cdots$ & $r^3\epsilon$ & $r^3\epsilon$ & $r^3\epsilon$ & $r^3\epsilon$ &$3-(n-1)r^3\epsilon-3\epsilon$\\
			\hline
			$\cdots$ & $\cdots$ & $\cdots$ & $\cdots$ & $\cdots$ & $\cdots$  & $\cdots$ & $\cdots$ & $\cdots$ & $\cdots$\\
			\hline
			$\mathbf{n-1}$ & $\epsilon$ & $\epsilon$ & $\epsilon$ & $\cdots$ & $\boxed{\epsilon}$ & $r^3\epsilon$ & $r^3\epsilon$ & $r^3\epsilon$ & $3-3r^3\epsilon-(n-1)\epsilon$\\
			\hline
			$\mathbf{n}$ & $\epsilon$ & $\epsilon$ & $\epsilon$ & $\cdots$ & $\epsilon$ & $\boxed{\epsilon}$ &$\epsilon$ & $\epsilon$ & $3-(n+2)\epsilon$\\
		\end{tabular}
    	\caption{Assume that item $e_{n+1}$ is allocated to one of the agents in $\{1,2\}$.}\label{tab:General-Case2b}
    \end{table}

	On the other hand, (no matter which agent receives item $e_{n+1}$) we have
	\begin{equation*}
	    v_1(X_1') \leq 2\epsilon, \quad 
	    v_2(X_2') \leq r\epsilon + \epsilon, \quad
	    v_3(X_3') = \epsilon.
	\end{equation*}
	Since only two items $e_{n+2}, e_{n+3}$ are not allocated, by Observation~\ref{observation:items-less-than-agents},
	the allocation is not $\gamma$-MMS.
    \end{proof}

\end{document}